\documentclass{elsarticle}
    \usepackage[standard,thmmarks]{ntheorem}
    \usepackage{hyperref,amsmath,xspace,tikz,cleveref}
    \renewtheorem{theorem}{Theorem}[section]
    \renewtheorem{proposition}[theorem]{Proposition}
    \renewtheorem{lemma}[theorem]{Lemma}
    \renewtheorem{corollary}[theorem]{Corollary}
    \theorembodyfont{\upshape}
    \renewtheorem{example}[theorem]{Example}
    \renewtheorem{remark}[theorem]{Remark}
    \renewtheorem{definition}[theorem]{Definition}
    \newcommand{\demph}{\textbf}
    \newcommand{\sdzero}{\texttt0}
    \newcommand{\sdone}{\texttt1}
    \newcommand{\str}{\mathbf}
    \newcommand{\length}[1]{|#1|}
    \newcommand{\NN}{\ensuremath{\mathbb N}\xspace}
    \newcommand{\timef}{\operatorname{time}}
    \newcommand{\p}{\ensuremath{\mathbf P}\xspace}
    \newcommand{\opt}{\ensuremath{\mathbf{OPT}}\xspace}
    \newcommand{\spt}{\ensuremath{\mathbf{SPT}}\xspace}
    \newcommand{\mpt}{\ensuremath{\mathbf{MPT}}\xspace}
    \newcommand{\B}{\ensuremath{\mathcal B}\xspace}
    \newcommand{\bff}{\ensuremath{\mathbf{BFF}}\xspace}
    \newcommand{\pf}{\mathrm{P}}
    \renewcommand{\S}{\Sigma^*}
    \newcommand{\e}{\mathtt{\epsilon}}
    \newcommand{\tup}[1]{\langle #1 \rangle}
    \newcommand{\onesec}[1]{#1_1}
    \newcommand{\twosec}[1]{#1_2}
    \newcommand{\Rec}{\mathcal{R}}
    \DeclareMathOperator*{\argmax}{arg\,max}
\newcounter{oldtocdepth}

\newcommand{\hidefromtoc}{%
  \setcounter{oldtocdepth}{\value{tocdepth}}%
  \addtocontents{toc}{\protect\setcounter{tocdepth}{-10}}%
}

\usepackage{hyperref}
\begin{document}

\title{Type-two Polynomial-time and Restricted Lookahead\tnoteref{t1,t2}}
\tnotetext[t1]{\copyright\ 2020. This manuscript version is made available under the CC-BY-NC-ND 4.0 license http://creativecommons.org/licenses/by-nc-nd/4.0/}
\tnotetext[t2]{{https://doi.org/10.1016/j.tcs.2019.07.003}}

\author[uvic]{Bruce M. Kapron\fnref{fn1}}
\ead{bmkapron@uvic.ca}
\author[inria]{Florian Steinberg\fnref{fn2}}
\ead{florian.steinberg@inria.fr}

\fntext[fn1]{Supported by an NSERC Discovery Grant}
\fntext[fn2]{Supported by the ANR project\emph{FastRelax}(ANR-14-CE25-0018-01) of the French National Agency for Research}
\address[uvic]{Computer Science Department, University of Victoria, Victoria, BC, Canada V8W 2Y2}
\address[inria]{INRIA Saclay, 91120 Palaiseau, France}

\begin{abstract}
    This paper provides an alternate characterization of type-two polynomial-time computability, with the goal of making second-order complexity theory more approachable.
    We rely on the usual oracle machines to model programs with subroutine calls.
    In contrast to previous results, the use of higher-order objects as running times is avoided, either explicitly or implicitly.
    Instead, regular polynomials are used.
    This is achieved by refining the notion of oracle-polynomial-time introduced by Cook.
    We impose a further restriction on the oracle interactions to force feasibility.
    Both the restriction as well as its purpose are very simple:
    it is well-known that Cook's model allows polynomial depth iteration of functional inputs with no restrictions on size, and thus does not guarantee that polynomial-time computability is preserved.
    To mend this we restrict the number of lookahead revisions, that is the number of times a query can be asked that is bigger than any of the previous queries.
    We prove that this leads to a class of feasible functionals and that all feasible problems can be solved within this class if one is allowed to separate a task into efficiently solvable subtasks.
    Formally put: the closure of our class under lambda-abstraction and application includes all feasible operations.
    We also revisit the very similar class of strongly polynomial-time computable operators previously introduced by Kawamura and Steinberg.
    We prove it to be strictly included in our class and, somewhat surprisingly, to have the same closure property.
    This can be attributed to properties of the limited recursion operator:
    It is not strongly polynomial-time computable but decomposes into two such operations and lies in our class.
\end{abstract}

\begin{keyword}
computability in higher types; feasible functionals; type-two polynomial time; oracle Turing machine; applied lambda-calculus; recursion on notation
\end{keyword}
\date{April 1, 2020}
\maketitle

\tableofcontents
\section{Introduction}
In the setting of ordinary computability theory, where computation is performed on finite objects (e.g., numbers, strings, or combinatorial objects such as graphs) there is a well-accepted notion of computational feasibility, namely poly\-no\-m\-ial-time computability.
The extended Church-Turing thesis codifies the convention: the intuitive notion of feasibility is captured by the formal model of computability by a polynomial-time Turing machine.\footnote{Ignoring issues related to randomization or quantum computing.}
From a programming perspective this can be interpreted as a formal definition of a class of programs that should be considered fast.
Of course this theory only applies to programs whose execution is determined from a finite string that is considered the input.
In practice, software often relies on external libraries or features user interaction.
One may address this by moving to a setting where a Turing machine does not only act on finite inputs but additionally interacts with \lq infinite inputs\rq.
This leads to the familiar \demph{oracle Turing machine} model, where infinitary inputs are presented via an oracle that can be fed with and will return finite strings, so that only finite information about the oracle function is available at any step of the computation.
The word oracle is used as no assumptions about the process to produce the values are made.
In particular, the oracle provides return values instantly.
From the software point of view this means judging the speed of a program independently of the quality of libraries or lazy users.
Since the oracle can be understood as type-one input and oracle machines to compute type-two functions, the investigation of resource consumption in this model is called second-order complexity theory.

Can a sensible account of feasible computation be given in this model?
If so, can it be kept consistent with the familiar notion of polynomial-time for ordinary Turing machine computation and the more traditional way of using oracle machines with $\sdzero$-$\sdone$ valued oracles \cite{Cook71}?
These problems were first posed by Constable in 1973 \cite{Constable73}.
He proposed only potential solutions and the task was taken up again by Mehlhorn in 1976, who gave a fully-formulated model \cite{MR0411947}.
This model is centered around Cobham's scheme-based approach to characterizing polynomial-time \cite{Cobham65}.
While such scheme based approaches are very valuable from a theoretical point of view \cite{MR2048061}, for some applications it may be desirable to have  a characterization that relies on providing bounds on resource consumption in a machine-based model.
Indeed, Mehlhorn related his formulation to the oracle machine model by proving that it satisfies the \demph{Ritchie-Cobham property}:
a functional is in his class if and only if  there is an oracle machine and a bounding functional from the class such that for any inputs the machine computes the value of the functional in question and has a running time bounded by the size of the bounding functional.
The impredicative nature of Mehlhorn's OTM characterization left open the possibility of a characterisation based more closely on the type-one model of polynomial-time Turing machines.

Only in 1996 did Kapron and Cook show that it is possible to give such a characterisation by relying on the notions of \demph{function length} and \demph{second-order polynomials} \cite{MR1374053}.
The resulting class of \demph{basic polynomial-time functionals} was proved equal to Mehlhorn's, providing evidence of its naturalness and opening the way for applications in diverse areas.
A representative but by no means exhasutive list includes work in computable analysis \cite{KawamuraC12}, programming language theory \cite{MR2295797}, NP search problems \cite{BeameCEIP98}, and descriptive set theory \cite{MR1058425}.
The model was also used as a starting point for understanding how complexity impacts classical results on computability in higher types \cite{MR1911553,MR2075336,MR1463765}.

Ideas similar to those used by Kapron and Cook were used for a number of logical characterizations of the basic polynomial-time functionals supporting appropriateness of the class.
Works using logics based on bounded arithmetic \cite{MR2053401,MR2103646} rely on implicit representations of second-order polynomials.
A drawback of the Kapron-Cook approach is that length functions and second-order polynomials are not particularly natural objects to work with.
For instance, the length of a function -- which can be viewed as the most basic example of a second-order polynomial -- is not feasible.
This has direct implications for the applications.
The most common approach to avoid technical difficulties is to restrict to length-monotone oracles \cite{KawamuraC12,MR1058425}.
This corresponds to using only a fragment of second-order complexity theory and may in turn lead to technical difficulties.

Additional support for Mehlhorn's class and insight in its structure came from initial doubts that it was broad enough to include all type-two functionals that should be considered feasible.
Cook formulated a notion of \demph{intuitive feasibility}, and pointed out that a type-two \demph{well quasi-ordering functional}, which meets the criteria of intuitive feasibility, is not in Mehlhorn's class \cite{MR1236005}.
Subsequent work uncovered a number of shortcomings of the notion of intuitive feasibility.
Seth provided a class that satisfies the conditions but has no recursive presentation \cite{SethRec} and also proved that Cook's functional does not preserve the Kalmar elementary functions \cite{MR1238294}.
Attempts by Seth and later by Pezzoli to formulate further restrictions on intuitively feasible functionals to avoid noted pitfalls lead back to Mehlhorn's class \cite{MR1727821}.

Cook's intuitive feasibility uses the notion of \demph{oracle polynomial-time}, which is formulated using ordinary polynomials.
A POTM (for \lq polynomial oracle Turing machine\rq) is an oracle machine whose running time is bounded in the maximum size of its string input and all answers returned by its oracle input during its computation on these inputs.
By itself, this notion is too weak to provide a class of feasible functionals: it is well known that iterating a polynomial-time function may result in exponential growth and that this is possible within this class.
While Cook's approach was to rule out this behaviour on a semantic level, an alternate approach explored by a number of works involves the introduction of further restrictions to the POTM model \cite{SethRec,MR1727821,MR1826285}.
Most of these restrictions are fairly elaborate and in some sense rely on bounding by second-order polynomials, at least implicitly.

The present paper investigates less elaborate ways to restrict the behaviour of POTMs.
We present two simple syntactic restrictions to the POTM model that give proper subclasses of Mehlhorn's class and prove them to -- when closed in a natural way -- lead back to the familiar class of feasible functionals.

The first restriction, originally introduced by Kawamura and Steinberg, is called \demph{finite length revision}, while operators computable by a POTM with finite length revision are called \demph{strongly polynomial-time computable} \cite{KawamuraS17}.
It is known that this class excludes very simple examples of polynomial-time computable operators.
The second restriction is similar, original to this work and we dub it \demph{finite lookahead revision}.
We call operators that are computable by such POTMs \demph{moderately polynomial-time computable}.
The name is motivated by our results that this class includes the strongly polynomial-time computable operators (Proposition~\ref{prop: spt contains optfilr}) and is contained in the polynomial-time operators (Proposition~\ref{resu:from step-count and lookahead revision number to running time}).
These inclusions are proven to be strict (Example~\ref{ex: filr not in p}), but in contrast to the case with strong polynomial-time it requires some effort to find a functional that is polynomial-time but not moderately polynomial-time.
Along the way we prove that in our setting an additional restriction on the POTMs that Kawamura and Steinberg impose is not actually a restriction (Lemma~\ref{resu:step-counts}).

In both cases, the failure to capture feasibility is due to a lack of closure under composition.
The main result of this paper (Theorem~\ref{th: main}) is that each of these classes, when closed under lambda-abstraction and application, results in exactly the polynomial-time functionals.
To prove this we establish moderate polynomial-time computability of limited recursion (Lemma~\ref{resu:R is mpt}) and provide a factorization of any moderately polynomial-time computable operator into a composition of two strongly polynomial-time computable operators (Theorem~\ref{resu:factorization}).
The proof of the later turns out to have a nice interpretation:
the outer operator executes the original machine while throwing exeptions in certain cases and the inner operator is an exception handler whose form only depends on restricted information about the original operator.
Finally, we point out a case where composition does not lead to a loss of moderate polynomial-time computability (Lemma~\ref{resu:composition}).

The notion of a POTM is likely what a person familiar with complexity theory would probably first propose if asked what programs with subroutine calls should be considered efficient.
The inadequacy of this model is very easy to grasp:
even if the subroutine is polynomial-time, there is no guarantee that the combined program runs in polynomial-time.
The two conditions of finite length and lookahead revision are very straightforward attempts to solve this issue.
We prove that imposing either in addition to the POTM condition leads to feasible programs and that it remains possible to produce solutions to all feasible problems as long as one is willing to separate the task at hand into subtasks if necessary.
We provide some --  but far from complete -- insight into when such a split is necessary and how it can be done.

\subsection*{Preliminaries}
Let $\Sigma$ denote any finite alphabet, and $\S$ the set of finite strings over $\Sigma$.
Usually $\Sigma=\{\sdzero,\sdone\}$, occasionally we use another separator symbol $\#$.
The empty string is denoted $\e$, and arbitrary elements of $\S$ are denoted $\str a, \str b, \dots$.
If $\str a,\str b \in \S$, we write $\str a\str b$ to denote their concatentation, $|\str a|$ to denote the length of $\str a$, $\str a^{\le n}$ to denote the prefix of $\str a$  of length $n \ge 0$.
We write $\str b \subseteq \str a$ to indicate that $\str b$ is an initial segment of $\str a$ (i.e. for some $0 \le n \le |\str a|$, $\str b = \str a^{\le n}$).
For every $k \in \mathbb{N}$ and $1 \le i \le k$ we note that there exist polynomial-time functions $\langle{\cdot,  \dots, \cdot}\rangle\colon (\S)^k \rightarrow \S$ and $\pi_{i,k}\colon\S \rightarrow \S$ such that $\pi_{i,k}(\tup{\str a_1,\dots,\str a_k})=\str a_i$.
We  assume that  for every $k$ there are constants $c_1,c_2$ such that $|\tup{\str a_1,\dots,\str a_k}|\le c_1\cdot(|\str a_1|+\dots+|\str a_k|)+c_2$ and that increasing the size of any of the strings $\str a_i$ does not decrease the size of the tuple.
The tupling functions are lifted to also operate on functions $\varphi_1,\ldots,\varphi_n\colon \S \to \S$ via $\langle \varphi_1,\ldots,\varphi_n\rangle(\str a) :=\langle\varphi_1(\str a),\ldots,\varphi_n(\str a)\rangle$.
A type 0 functional is an element of $\S$, and for $t \in \NN$, a \demph{type $t+1$ funtional} is a mapping from functionals of type $\le t$ to $\S$.
This paper is mostly concerned with type $t$ functionals for $t \le 2$.

\subsection{Second-order complexity theory}\label{sec:second-order complexity theory}
In \cite{MR1374053}, Kapron and Cook introduce a computational model for type-two polynomial time functionals using oracle Turing machines.
We begin by reviewing their model.
For notational simplicity, we do this in the operator setting:
Denote by $\B:= \S \to \S$ the Baire space, that is the collection of all univariate type 1 functions.
The elements of $\B$ are denoted by $\varphi$, $\psi$, \ldots.
An \demph{operator} is a mapping $F: \B \rightarrow \B$.
Note that each operator $F$ can be assinged a functional $\mathcal F$ via $\mathcal F (\varphi, \str a) := F(\varphi)(\str a)$.
Conversely, any total type 2 functional can be translated to an operator by use of the pairing functions.
However, it should be kept in mind that in the absence of totality or pairing functions the notions of operators and functionals may diverge.

        An \demph{oracle Turing machine} (OTM) or for short \demph{oracle machine} is a Turing machine that has distinguished and distinct \demph{query} and \demph{answers} tapes and a designated \demph{oracle state}.
        The run of an oracle machine $M$ on oracle $\varphi \in \B$ and input $\str a$ proceeds as the run of a regular machine on input $\str a$, but whenever the oracle machine enters the oracle state, with $\str b$ written on the query tape, $\varphi(\str b)$ is placed immediately on the answer tape, and the read/write head returns to its initial position on both of these tapes.
        If the machine terminates we denote the result by $M^\varphi(\str a)$.
        Note that each oracle machine naturally computes a partial functional of type $\B \times \S \to \S$ and may also be considered to compute a partial operator via $F_M(\varphi) = \psi \iff \forall \str a, M^\varphi(\str a) = \psi(\str a)$.

        The number of steps $\timef_{M}(\varphi,\str a)$ an oracle machine $M$ takes given oracle $\varphi$ and input $\str a$ is counted as in a Turing machine with the following addition already implied above:
        entering the oracle state takes one time step, but there is no cost for receiving an answer from the oracle\footnote{In this paper, we follow the {\em unit-cost} model \cite{MR0411947}, as opposed to the {\em length-cost} model \cite{MR1374053}.
        Note that while an answer is received from the oracle in a single step, any further processing takes time dependant on its length.}.
        The running time of an oracle machine usually depends on the oracle.

        To be able to talk about bounds for this running time it is necessary to have a notion for the size of an oracle. The idea behind this definition, originally given in \cite{MR1374053}, is that the cost of a function input should itself be a function which on input $n \in \NN$ gives the maximum output size over all inputs of size at most $n$.
        \begin{definition}
            For a given $\varphi\in \B$ define its \demph{size function} $\length\varphi\colon \NN \to\NN$ by
            \[ \length\varphi(n) := \max_{\length{\str a} \leq n} \{\length{\varphi(\str a)}\}. \]
        \end{definition}
        This suggests the type $\NN^\NN \times \NN \to \NN$ as the right type for running times:
        If $T$ is a function of this type, we say that the running time of an oracle machine $M$ is bounded by $T$ if for all oracles $\varphi:\S \rightarrow \S$ and all strings $\str a$ it holds that
        \[ \timef_{M}(\varphi,\str a) \leq T(\length \varphi, \length{\str a}). \]
        The only thing left to do is to pick out the time bounds that should be considered polynomial.
        \begin{definition}
            The set of \demph{second-order polynomials} is the smallest subset of $\NN^\NN\times\NN\to\NN$ that contains the functions $(l,n)\mapsto 0$, $(l,n)\mapsto 1$, $(l,n)\mapsto n$, is closed under point-wise addition and multiplication and such that whenever $P$ is an element, then so is $(l,n)\mapsto l(P(l,n))$.
        \end{definition}
        Thus, by this definition each second-order polynomial can be represented by a term in a variable symbol $n$ for numbers, a symbol $l$ for a function, binary addition and multiplication.
        For instance the function defined by $P(l,n) := n \cdot l(n^2+1) \cdot l(n^5 + l(n +1)) + 1$ is a second-order polynomial.
We may now use Kapron and Cook's characterization \cite{MR1374053} as our definition of Mehlhorn's class:

\begin{definition}\label{def:BPT}
An operator $F\colon\B\to\B$ is \demph{polynomial-time compu\-table} if there is an oracle machine $M$ and a second-order polynomial $P$ such that for all oracles $\varphi$ and all strings $\str a$ it holds that
\[
    F(\varphi,\str a)=M^\varphi(\str a) \quad\text{and}\quad \timef_{M}(\varphi,\str a) \leq P(\length \varphi, \length{\str a})
\]
We use \p to denote the class of all polynomial-time operators.
\end{definition}
This notion gives rise to a notion of polynomial-time computable functionals of type two.
By abuse of notation we also refer to this class of functionals by $\p$.
The functional view becomes important in Section~\ref{sec:lambda}.
In the literature these notions are often referred to as \lq basic polynomial-time\rq.
As discussed in the introduction this is due to past uncertainties about the class being broad enough.
We believe that enough evidence has been gathered that the class is appropriate and drop the \lq basic\rq.

       Consider the following result taken from \cite{KawamuraS17} that implies the closure of polynomial-time computable operators under composition.
       \begin{theorem}\label{resu: closure under composition}
         Let $P$ and $Q$ be second-order polynomials that bound the running times of oracle machines $M$ and $N$.
         Then there exists an oracle machine $K$ that computes the composition of the operators corrsponding to $M$ and $N$, i.e. for all $\varphi$ and $\str a$ it holds that $K^{\varphi}(\str a) = N^{\lambda \str b. M^{\varphi}(\str b)}(\str a)$.
         Furthermore $K$ can be chosen such that forall $\varphi$ and $\str a$,
            \[ \timef_{K}(\varphi, \str a) \leq C (P(\length\varphi, Q(\lambda n.P(\length\varphi,n),\length{\str a})) \cdot Q(\lambda n. P(\length\varphi,n),\length{\str a}) + 1). \]
        \end{theorem}
        The proof is straightforward and the reader not familiar with the setting may sketch a proof to get a feeling for oracle machines, higher-order time bounds and second-order polynomials.
        Unsurprisingly, the machine $K$ in this proof is constructed by replacing the oracle query commands in the program of the machine computing the outer operator by copies of the program of the machine computing the inner operator and slightly adjusting the rest of the code.
        Note how this result lends itself to generalizations:
        Kawamura and Steinberg use it to lift closure under composition to a class of partial operators that they still refer to as polynomial-time computable.
        The proof also remains valid if the second-order polynomials $P$ and $Q$ are replaced by arbitrary functions $S$ and $T$ that are monotone in the sense that whenever $l$ is pointwise bigger than $k$ and both are non-decreasing then the same holds for the functions $\lambda n. T(l,n)$ and $\lambda n.T(k,n)$.

    Reasoning about second-order polynomials as bounding functions can at times be tricky.
    Their structure theory is significantly less well developed than that of regular polynomials.
    Indeed, it is not clear whether second-order polynomials allow a nice structure theory at all.
    Furthermore, the use of nonfinitary objects in running times raises the question of computational difficulty of evaluating such bounds.
    It is a very simple task to find the length of a string from a string.
    In contrast, evaluating the length $\length\varphi (n) = \max_{\length{\str a}\leq n}\length{\varphi(\str a)}$ of a string function is intuitively a hard task as it involves taking a maximum over an exponential number of inputs.
    The following theorem from \cite{MR1374053} makes this intuition formal.

    \begin{theorem}
        The length function is not polynomial-time computable:
        An operator $L$ that fulfills
        \[ \length{L(\varphi)(\str a)} = \length{\varphi}(\length{\str a}) \]
        cannot be polynomial-time computable.
    \end{theorem}
    As a consequence, a running time bound of an oracle machine is not very useful for estimating the time of a run on a given oracle and input.
    Even if the running time is a second-order polynomial $P$, to get the value $P(\length\varphi,\length{\str a})$ one has to evaluate the length function several times.

    Of course, in this setting the task is a little silly.
    It is possible to evaluate the machine and just count the number of steps it takes.
    This results in a tighter bound that can be computed from the oracle and the input in polynomial time.
    However, from a point of view of clockability, the problem is relevant:
    given a second-order polynomial $P$ that is interpreted as a \lq budget\rq\ and an oracle machine $M$ that need not run in polynomial time it is in general impossible to specify another machine $N$ that runs in polynomial-time and such that for all oracles and inputs
    \[ \timef_{M}(\varphi,\str a)\leq P(\length{\varphi},\length{\str a}) \quad\Rightarrow\quad N^\varphi(\str a) = M^\varphi(\str a). \]
    That is: $N$ returns the correct value in the case that the run of $M$ is in budget \cite{KawamuraS17}.

\subsection{Oracle polynomial-time}

        The following notion was originally introduced by Cook \cite{MR1236005} and has been investigated by several other authors as well \cite{SethRec,MR1238294,MR1727821}.
        Recall that for an oracle machine $M$ the number of steps this machine takes on input $\str a$ with oracle $\varphi$ was denoted by $\timef_{M}(\varphi,\str a)$ and that for counting the steps, the convention to count an oracle query as one step was chosen.
        \begin{definition}\label{def:opt}
            Let $M$ be an oracle machine.
            For any oracle $\varphi$ and input $\str a$ denote by $m_{\varphi,\str a}$ the maximum of the lengths of the input and any of the oracle answers that the machine gets in the run on input $\str a$ with oracle $\varphi$.
            The machine $M$ is said to run in \demph{oracle polynomial-time} if there is a polynomial $p$ such that for all oracles $\varphi$ and inputs $\str a$
            \[ \tag{sc}\label{eq:step-count} \timef_{M}(\varphi,\str a) \leq p(m_{\varphi,\str a}). \]
            Let $\opt$ denote the class of operators that are computed by a machine that runs in oracle polynomial-time.
        \end{definition}
        To avoid confusion with different notions of running times, we call a function $t\colon\NN\to\NN$ such that $\timef_M(\varphi,\str a) \leq t (m_{\varphi,\str a})$ for all $\varphi$ and $\str a$  a \demph{step-count} of $M$.
        That is: a step-count fulfills the condition from \eqref{eq:step-count} but need not be a polynomial and an oracle machine runs in oracle polynomial-time if and only if it has a polynomial step-count.
        An oracle machine with a polynomial step-count may be referred to as a POTM.

        The nature of the restrictions imposed on POTMs significantly differs from imposing a second-order time bound as is done in Definition~\ref{def:BPT}.
        Instead of using higher-order running times, the same type of function that are used as running times for regular Turing machines is used.
        The dependence on the oracle is accounted for by modifying the input of the bounding function.
        This appears to be a relaxation of bounding by second-order polynomials.
        The following result of \cite{MR1236005} shows that this is indeed the case.

        \begin{theorem}[$\p\subseteq\opt$]\label{resu:polytime implies opt}
            Any oracle machine that runs in time bounded by a second-order polynomial has a polynomial step-count.
        \end{theorem}

        \begin{proof}
            Let $M$ be an oracle machine that runs in time bounded by a second-order polynomial $P$.
            For $n\in\NN$ let $l_n:\NN\to\NN$ be the constant function with value $n$.
            Claim that the polynomial given by
            \[ p(n):= P(l_n,n) \]
            is a step-count.
            To verify this let $\varphi$ be an arbitrary oracle and $\str a$ an arbitrary input.
            Replace the oracle $\varphi$ with the oracle $\psi$ that returns the same values as $\varphi$ on all of the oracle queries that are asked in the computation of $M^\varphi(\str a)$ and returns the empty string on all other arguments.
            Since the machine can only interact with the oracles through the queries, the computations $M^\varphi(\str a)$ and $M^\psi(\str a)$ coincide.
            In particular the time these computations take are the same.
            By definition of $\psi$, $l_n$ and $m_{\varphi,\str a}$ it holds that $\length{\psi}\leq l_{m_{\varphi,\str a}}$.
            It follows from  $P$ being a running time bound of $M$ and the monotonicity of second-order polynomials that
            \[
                \timef_{M}(\varphi,\str a)=\timef_{M}(\psi,\str a)\leq P(\length{\psi},\str a) \leq P(l_{m_{\varphi,\str a}},m_{\varphi,\str a}) = p(m_{\varphi,\str a}).
            \]
            Since $\varphi$ and $\str a$ were arbitrary, it follows that $p$ is a polynomial step-count of $M$.
        \end{proof}

        It is well known that \p forms a proper subclass of \opt.
        There exist operators in \opt that do repeated squaring and thus do not preserve polynomial-time computability.
        \begin{example}[$\p\subsetneq \opt$]\label{ex:repeated squaring}
            The operator
            \[ F(\varphi)(\str a):= \varphi^{\length{\str a}}(\sdzero) \]
            can be computed by a machine that runs in oracle polynomial-time but does not preserve polynomial-time computability as it maps the polynomial-time computable function $\varphi(\str a) :=\str a\str a$ to a function whose return values grow exponentially.
        \end{example}

\section[Recovering feasibility from OPT]{Recovering feasibility from \opt}
The failure of \opt to preserve polynomial-time computability indicates that it is unsuitable as a class that captures an acceptable notion of feasibility.
We may ask whether there is a natural way to restrict the POTM model to recover feasibility.
One way to do this is to introduce preservation of polynomial-time functions as an {\em extrinsic} or a {\em semantic} restriction \cite{MR1251285}.
This is the approach taken by Cook with his notion of intuitively feasible functionals \cite{MR1236005}.
Since the formulation of Cook's restrictions is most comfortably done using lambda calculus we postpone restating them to Section~\ref{sec:lambda}.
Here, we consider {\em intrinsic} or {\em syntactic} restrictions of the POTM model instead.
While in part this is motivated by some of the drawbacks of the extrinsic approach, we believe that the syntactic approach stands on its own merit.
In particular, if the syntactic condition is simple enough and checkable with minimal overhead, it provides simpler analysis techniques for showing that a particular operator is feasible.

Motivated by the difficulty encountered with repeated squaring in Example~\ref{ex:repeated squaring}, we consider POTMs with restrictions on the oracle access that disallow this behaviour.
Similar restrictions have been considered by Seth \cite{SethRec,MR1238294}.
Seth's class $C_0$ consists of operators computable by POTMs whose number of queries to the oracle is uniformly bounded by a constant, while his class $C_1$ is defined using a form of dynamic bound on the size of any query made to the oracle.
It should be noted, that Seth proves his class $C_1$ to coincide with Mehlhorn's class $\p$ and that it employs POTMs that are clocked with something that faintly resembles second-order polynomials.
The class $C_0$ is too restrictive to allow the recovery of all polynomial-time operators.
The class $C_1$ is aimed at finding a bigger class of operators that should still be considered feasible and thus overly complicated for our purposes.
For the similar reasons we do also not go into detail about restrictions considered by Pezzoli \cite{MR1727821}.

We seek conditions that do not beg the question of whether $\p$ is a maximal class of feasible functionals, disallow unbounded iteration (as in the example of repeated squaring,) and yet are less restrictive than those of $C_0$.
\subsection{Strong polynomial-time computability}

The first restriction on \opt that we consider was introduced by Kawamura and Steinberg in \cite{KawamuraS17}:
        \begin{definition}
            An oracle machine is said to run with \demph{finite length revision} if there exists a constant $r$ such that in the run of the machine on any oracle and any input the number of times it happens that an oracle answer is bigger than the input and all of the previous oracle answers is at most $r$.
        \end{definition}

        It should be noted that Kawamura and Steinberg use a slightly different notion of a step-count.
        They say that a function $t\colon\NN\to\NN$ is a \demph{step-count} of an oracle machine $M$ if for all oracles and inputs it holds that
        \[ \forall k \in \NN\colon k\leq \timef_{M}(\varphi,\str a) \Rightarrow k\leq t(m_{k,\varphi,\str a}), \]
        where $m_{k,\varphi,\str a}$ is the maximum of $\length{\str a}$ and the biggest oracle answer given in the first $k$ steps of the computation of $M$ with oracle $\varphi$ and input $\str a$.
        For the function $t$ to be a step-count in the sense of the present paper it suffices to satisfy the condition for the special choice $k:= \timef_{M}(\varphi,\str a)$.
        The reason we use the same name for both of these notions is that they are equivalent in our setting.
        The result is interesting in its own right as the advantage of Kawamura and Steinberg's notion is that it can be checked on the fly whether a given polynomial is a step-count of an oracle machine without risking to spend a huge amount of time if this is not the case.
        We only state the equivalence for the case we are really interested in, but the proof generalizes.

        \begin{lemma}\label{resu:step-counts}
            Every oracle machine that computes an operator from \opt has a polynomial step-count (in the sense of Kawamura and Steinberg).
        \end{lemma}

        \begin{proof}
            Let $M$ be a machine that computes an element of \opt.
            Then there exists a polynomial $p$ such that
            \[ \forall \varphi,\str a: \timef_{M}(\varphi,\str a) \leq p(m_{\varphi,\str a}). \]
            We claim that this polynomial is already a step-count in the sense of Kawamura and Steinberg.
            Let $m_{k, \varphi, \str a}$ be the maximal value that is returned by the oracle in the computation of $M$ on input $\str a$, oracle $\varphi$ and before the machine takes the $k$-th step.
            Towards a contradiction, assume that $p$ was not a step-count.
            Then there exists some time $k$, a string $\str a$ and an oracle $\varphi$ such that $p(m_{k, \varphi,\str a}) > k$.
            Replace the oracle $\varphi$ by the oracle $\psi$ that returns the exact same answers as $\varphi$ on the strings that are asked in this run before the $k$-th step and returns $\epsilon$ on all other strings.
            Note that due to the machine being deterministic the computation of $M$ on oracle $\varphi$ and $\psi$ are identical up to the $k$-th time step.
            Furthermore, due to the definition of $\psi$, the number $m_{k,\varphi,\str a}$ coincides with the biggest size of any oracle answer $\psi$ gives.
            Thus, it follows that $m_{\psi,\str a} = m_{k,\varphi,\str a}$ and therefore
            \[ \timef_{M}(\psi,\str a) \geq k > p(m_{k,\varphi,\str a}) = p(m_{\psi,\str a}). \]
            This contradicts the assumption about $p$.
        \end{proof}

        The key idea of the above proof is that the oracle can be modified arbitrarily.
        In a setting where not all oracles are eligible this might not be possible anymore.
        In this case it is advisable to work with step-counts in the sense of Kawamura and Steinberg.
        Since this paper only considers total operators, i.e.\ no restrictions are imposed on the oracles, it is irrelevant which notion is used.
        In particular we may formulate strong polynomial-time computability as introduced in \cite{KawamuraS17}.
        \begin{definition}
            An operator is \demph{strongly polynomial-time computable} if it can be computed by an oracle machine that has both finite length revision and a polynomial step-count.
            The class of these operators is denoted by $\spt$.
        \end{definition}

        As the name suggests, strong polynomial-time compuatbility implies polynomial-time computability.
        We state this as it can be deduced from the results of this paper.
        A direct proof is given in \cite{KawamuraS17}.

        \begin{proposition}[$\spt \subseteq \p$]\label{resu:from step-count and length revision number to running time}
           The running time of a machine that has a polynomial step-count and finite length revision can be bounded by a second-order polynomial.
        \end{proposition}

        \begin{proof}
            This is an immediate consequence of Proposition~\ref{prop: spt contains optfilr} that proves the inclusion of $\spt$ in a broader class called $\mpt$ that is introduced in the next section and proven to be included in $\p$ in Proposition~\ref{resu:from step-count and lookahead revision number to running time}.
        \end{proof}

        A merit of strong polynomial-time computability is that it has a direct interpretation as additional information about the running time of the machine:
        Knowing a polynomial step-count of a program and the number $r$ of length revisions, one can modify the program to provide real-time information about how long it estimates it will run.
        It can provide an estimate of the remaining computation time under the assumption that all necessary information has already been obtained from the oracle.
        In case new information is gained via oracle interaction it may update this estimate, but it may only do so at most $r$ times.

        A drawback of strong polynomial-time computability is that it severely restricts the access a machine has to certain oracles.
        It may always run into an increasing sequence of answers early in the computation.
        Once it runs out of length revisions, it can not pose any further oracle queries.
        There is no way to design a machine with finite length revision that does not simply abort when the revision budget is exceeded.
        This is reflected in the following example that shows that there are operators from $\p$ that are not in $\spt$, and was first considered in \cite{MR0411947}.

        \begin{example}[$\spt \subsetneq \p$]\label{ex:not spt}
            The operator
            \[ F(\varphi)(\str a):= \sdone^{\max_{\str b\subseteq \str a}\length{\varphi(\str b)}} \]
            is polynomial-time computable but not strongly polynomial-time.
            Any machine computing $F$ must query $\varphi$ at every $\str b \subseteq \str a$.
            Regardless of the order in which the machine decides to ask the queries, there is always an oracle whose answers are increasing in size.
            However, $F \in \p$, as it may be computed by examining $|\str a|$ queries, each of which is of size at most $|\varphi|(|\str a|)$.
        \end{example}
        The idea behind the counterexample is that it is possible to construct an oracle that forces an arbitrary number of length revisions for a fixed machine and polynomial step-count.
        More details for very similar examples can be found in \cite{KawamuraS17} or \Cref{ex:repeated squaring revisited} and the same method is also used in Example~\ref{ex: filr not in p}.
\subsection{Finite lookahead revision}
       The notion of strong polynomial-time computability rests on controlling the size of answers provided by calls to the oracle.
       While this restriction achieves the goal of disallowing anything but finite depth iteration, Example~\ref{ex:not spt} shows that it also excludes rather simple polynomial-time computable operators.
       This suggests an alternate form of control, namely controlling the size of the queries themselves instead of the answers.
        \begin{definition}
            An oracle machine is said to run with \demph{finite lookahead revision} if there exists a natural number $r$, such that for all possible oracles and inputs it happens at most $r$ times that a query is posed whose size exceeds the size of all previous queries.
        \end{definition}
        We are mostly interested in operators that can be computed by a machine that both has finite lookahead revision and a polynomial step-count.
        In keeping with the terminology of strong polynomial time, we shall call the class of operators so computable \demph{moderate polynomial time}, denoted \mpt.
        \begin{example}[$\mpt \not\subseteq \spt$]
          Consider the operator $F$ that maximizes the size of the return value of the oracle over the initial segments of the string input.
          This operator was used to separate $\spt$ from $\p$ in Example~\ref{ex:not spt} and therefore fails to be strongly polynomial-time computable.
          The operator $F$ belongs to $\mpt$:
          A machine computing $F$ on inputs $\varphi,\str a$ may just query the initial segments of $\str a$ in decreasing order of length to obtain the maximum answer.
        \end{example}

        While the definition above seems reasonable enough, it entails that machines sometimes need to unnecessarily pose oracle queries:
        to ask all interesting queries up to a certain size it may be necessary to pose a big query whose answer is of no interest to the computation just to avoid lookahead revisions during the computation.
        It is possible to tweak the oracle access of machines to avoid this behaviour and still capture the class of operators that are computed with finite lookahead revision.
        For instance one may use a data structure that may be viewed as finite stack of unbounded oracle tapes. The structure is initialized with some constant number of tapes, and only popping is allowed. Each pop requires the machine to specify in unary a number of cells to which the tape that is being popped is truncated.
        While this model is not the most straightforward one, it is appealing since similar restrictions on the oracle access have to be imposed to reason about space bounded computation in the presence of oracles \cite{MR973445,MR3259646}.

        \begin{proposition}[$\mpt\subseteq \p$]\label{resu:from step-count and lookahead revision number to running time}
          The running time of a machine that has a polynomial step-count and finite lookahead revision can be bounded by a second-order polynomial.
        \end{proposition}

        \begin{proof}
            Let $M$ be a machine with polynomial step-count $p$ that never does more than $r$ lookahead revisions.
            Since increasing a function point-wise preserves being a step-count, we may assume $p(n) \geq n$.
            To see that the running time is bounded by the second-order polynomial
            \[ P(l,n) := (p\circ l)^r(p(n)) + p(n) \]
            prove the stronger statement that the running time of such a machine is bounded by the non-polynomial function $T(l,n):=\max\{(p\circ l)^r(p(n)),p(n)\}$ by induction over the lookahead revision number $r$.
            If $r$ is zero, then the machine does not ask any oracle queries and has to terminate within $p(\length{\str a})$ steps.
            If the assertion has been proven for $r$ and $M$ is a machine with lookahead revision number $r+1$, then consider the machine $\tilde M$ that carries out the same operations as $M$ does but aborts in cases where $M$ attempts to issue the oracle query that leads to the $(r+1)$-th lookahead revision.
            The machine $\tilde M$ has lookahead revision number $r$ and $p$ as a step-count.
            Thus, by the induction hypothesis, the oracle query that leads $\tilde M$ to abort and that triggers the last lookahead revision in $M$ can at most have been of length $\max\{(p\circ \length{\varphi})^r(p(\length{\str a})),p(\length{\str a})\}$.
            By definition of the length function the answer of the oracle has at most length $\length{\varphi}$ applied to that value.
            Since all later oracle queries of $M$ have to be of smaller length and since $p$ is a step-count of the machine, the time that $M$ takes can be bounded by the step-count applied to the maximum of that value and the length of the input.
            Recall that we assumed $p(n)\geq n$ and therefore $p(\length{\varphi}(p(n)))\geq \length{\varphi}(n)$.
            Thus,
            \[
                p\Big(\max\Big\{\length{\varphi}\Big((p\circ\length\varphi)^r\big(p(\length{\str a})\big)\Big),\length\varphi\big(p(\length{\str a})\big),\length{\str a}\Big\}\Big) =\max\big\{(p\circ\length{\varphi})^{r+1}\big(p(\length{\str a})\big),p(\length{\str a})\big\}.
            \]
            Which proves the assertion, and therefore also the proposition.
        \end{proof}

        \begin{proposition}[$\spt\subseteq \mpt$]\label{prop: spt contains optfilr}
            Every strongly polynomial-time computable operator is moderately polynomial-time computable.
        \end{proposition}

        \begin{proof}
            let $M$ be a machine that proves that an operator is strongly polynomial-time computable.
            Let $p$ be a polynomial step-count of the machine and let $r$ be its number of length revisions.
            Consider the machine $N$ that works as follows:
            First it checks whether $p(\length{\str a})$ is bigger than $\length{\str a}$ and if so poses an oracle query of this size.
            This leads to a lookahead revision.
            Then it follows the first $p(\length{\str a})$ steps that $M$ takes while remembering the size of the biggest oracle answer that it receives.
            This does not lead to a lookahead revision as the machine $M$ does not have enough time to formulate a query big enough.
            Since $p$ is a step count, $M$ can at most take $p(\length{\str a})$ steps before it has to either terminate or encounter a length revision.
            If it terminates, let $N$ return its return value.
            If it encounters a length revision, then the maximum $m$ of the return values that $N$ recorded is bigger than $\length{\str a}$ and $N$ repeats the proceedure with $p(\length{\str a})$ replaced by $p(m)$.
            Note that the whole process can at most be iterated $r$ times as each time $M$ has to either terminate or encounter a length revision.
            Thus, in the last repetition, $M$ has to terminate and it follows that $M$ and $N$ compute the same operator.
            Also note that during each of the repetitions $N$ encounters exactly one lookahead revisions and thus the number of lookahead revisions of $N$ is bounded by the number of length revision of $M$.
            Finally $N$ has a polynomial step-count as the number of steps required to carry out the operations described above can easily be bounded by a polynomial.
            It follows that $N$ proves that the operator computed by $M$ belongs to $\mpt$.
        \end{proof}
        Although $\mpt$ is more powerful than $\spt$, it is still not powerful enough to capture all of $\p$:
        \begin{example}[$\mpt \subsetneq \p$]\label{ex: filr not in p}
            Consider the operator $F:\B\to\B$ defined as follows:
            First recursively define a sequence of functions $F_i:\B \to \Sigma^*$ by
            \[ F_0(\varphi) := \epsilon \quad\text{and}\quad F_{n+1}(\varphi) := \big((\varphi\circ \varphi)(F_{n}(\varphi))\big)^{\leq \length{\varphi(\epsilon)}}. \]
            That is: start on value $\epsilon$ and iterate $n$ times the process of applying the function $\varphi\circ \varphi$ and then truncating the result to have length $\length\varphi(0)$.
            Set
            \[ F(\varphi)(\str a) := F_{\length{\str a}}(\varphi). \]
            We claim that this operator is polynomial-time computable but can not be computed by a machine that has finite lookahead revision and a polynomial step-count.
            To see that $F$ is polynomial-time computable note that for any string $\str a$ of length bigger than one we have
            \[ \length{F(\varphi)(\str a)} = \length{(\varphi\circ \varphi)(F_{\length{\str a}-1}(\varphi)(\str a))^{\leq \length{\varphi(\epsilon)}}} \leq \length{\varphi(\epsilon)} = \length{\varphi}(0). \]
            Thus, the straight-forward algorithm that on input $\str a$ and oracle $\varphi$ computes the sequence $F_0(\varphi),\ldots,F_{\length{\str a}}(\varphi)$ runs in time about $\length{\str a} \cdot (\length{\varphi}(0) +\length{\varphi}(\length{\varphi}(0)))$.

            To see that $F$ cannot be computed by a machine with polynomial step-count and finite lookahead revision, let $M$ be a machine that computes $F$ and has a polynomial step-count $p$.
            Without loss of generality assume $p(n)\geq n$.
            For any given number $k\geq 1$ construct a sequence of oracles $\psi_0,\ldots, \psi_k$ such that $\psi_i$ forces $i$ lookahead revisions on input $\sdone^k$.
            To do so, first choose some $m > p(k)$ such that $(p+1)^k(m)< 2^{m}-k-2$.
            This is fulfilled by almost all natural numbers as the right hand side grows exponentially while the left hand side is a polynomial, in particular an appropriate $m$ exists.

            Recursively define the sequence $\psi_i$ in parallel with a sequence $\str a_j$ of pairwise distinct strings of length $m$ such that each string has at least one digit that is $\sdzero$.
            Let $\psi_0$ be the constant function returning the empty string and for $i\geq 1$ define $\psi_i$ from the strings $\str a_1,\ldots,\str a_i$ by
            \[ \psi_i(\str b):=\begin{cases} \str a_1 & \text{if }\str b = \epsilon \\
            \str a_{j} &\text{if } \str b = \sdone^{(p+1)^{j-1}(m)} \text{ for some } 2\leq j \leq i \\ \sdone^{(p+1)^{j-1}(m)} &\text{if } \str b = \str a_j\text{ for some }j\leq i \\
            \epsilon &\text{otherwise.}\end{cases} \]
            The function $\psi_i$ is well-defined due to the assumptions about the sequence $\str a_j$.

            The string $\str a_j$ is recursively defined from $\psi_0,\ldots,\psi_{j-1}$ as follows:
            For $j=1$ recall that $\psi_0$ was defined to be the constant function returning the empty string.
            Consider the computation $M^{\psi_0}(\sdone^k)$.
            Since $p$ is a step-count and $p(k)< m \leq 2^m-2$ (the second inequality follows since $m\geq2$ due to the assumptions), there exists at least one string $\str a_1$ of length $m$ that neither coincides with $M^{\psi_0}(\sdone^k)$ nor $\sdone^m$ nor with any of the oracle queries asked in this computation.
            Now assume that all $\str a_{j'}$ with $j'\leq j<k$ have been defined.
            This means that $\psi_j$ is defined.
            Consider the computation of $M^{\psi_j}(\sdone^k)$.
            Since the length of the return values of $\psi_j$ is bounded by $(p+1)^j(r)$, the number of steps in this compuation is smaller than
            \[ p((p+1)^j(k))\leq (p+1)^{j+1}(k)<2^m-k-2. \]
            Thus, there exists a string $\str a_{j+1}$ of length $m$ that is different from the $j+2<k+2$ strings $\str a_1,\ldots,\str a_j$, $\sdone^m$ and $M^{\psi_j}(\sdone^k)$ as well as all the oracle queries that are asked in the computation.

            This finishes the construction and it is left to prove that $\psi_i$ forces $i$ lookahead revisions.
            The proof proceeds by induction on $i$.
            The case $i=0$ is trivial as forcing $0$ lookahead revisions does not require doing anything.
            Next assume that the assertion has been proven for $i<k$.
            Claim that the string $\sdone^{(p+1)^i(m)}$ is posed as a query in the computation $M^{\psi_{i+1}}(\sdone^k)$.
            Since the values of $\psi_i$ and $\psi_{i+1}$ only differ in the strings $\str a_{i+1}$ and $\sdone^{(p+1)^i(m)}$, the computations of $M$ on $\psi_i$ and $\psi_{i+1}$ coincide up to the point where either of these strings is posed as a query.
            Since the computation with oracle $\psi_i$ takes at most
            \[ p(\max_{\str a\in\Sigma^*}\{\length{\psi_i(\str a)},r\}) = p((p+1)^{i-1}(m)) < (p+1)^i(m) \]
            steps, the machine either queries $\str a_i$ or none of the two.
            Towards a contradiction assume that the machine poses neither of the two queries.
            The runs on $\psi_{i+1}$ and $\psi_i$ coincide and -- by the construction of the string $\str a_i$ -- the return value is different from $\str a_i$.
            Therefore, the run on the oracle $\psi'$ which is identical to $\psi_{i+1}$ everywhere but on input $\sdone^{(p+1)^i(m)}$ where it returns $\str a_i$ instead of $\str a_{i+1}$ leads to an identical run.
            The fact that $F(\psi')(\sdone^k) = \str a_i$ leads to a contradiction with the assumption that $M$ computes $F$.
            Thus, $\str a_i$ has to be posed as an oracle query in the computation $M^{\psi_{i+1}}(\sdone^k)$.
            The same argument with the same counter example function $\psi'$ also proves that the string $\sdone^{(p+1)^i(m)}$ has to be posed as a query.
            Since no query before the query $\str a_i$ can have had size $(p+1)^i(m)$, posing this query requires an additional lookahead revision.
        \end{example}

    \subsection{Existence of step-counts and an application}

        While the notion of finite lookahead revision is most useful in the presence of a polynomial step-count, it is also meaningful for general operators.
        Consider the iteration operator $ F(\varphi)(\str a)= \varphi^{\length{\str a}}(\sdzero)$ from \Cref{ex:repeated squaring}: intuitively it should be possible to prove that this operator cannot be computed by a machine with finite lookahead revision at all.
        However, the usual construction of a counterexample oracle only generalizes under the assumption that the machine considered has a step-count.
        Luckily, using topological methods that originate from computable analysis, it is possible to prove that step counts do always exist:
        In \cite{MR2090390}, Schr\"oder notes that a machine that computes a total operator has a well-defined running time.
        In the setting of regular machines the existence of a function $t:\NN\to\NN$ that bounds the running time is obvious: the set of all strings of length smaller than a fixed natural number $n$ is finite and a time bound of a machine without an oracle can therefore be obtained by taking the maximum of the times of its runs on each of these.
        Recall that the size $\length \varphi\colon \NN\to\NN$ of an oracle is defined by
        \[ \length\varphi(n) := \max_{\length{\str a} \leq n} \{\length{\varphi(\str a)}\}. \]
        The sets $K_l$ of all oracles whose size is bounded by a given function $l\colon \NN \to\NN$ is almost never finite.
        However, if Baire space is given its standard topology, all of the sets $K_l$ are compact.
        Schr\"oder notes that for any machine $M$, the time function $(\varphi,\str a) \mapsto \timef_M(\varphi,\str a)$ is a continuous function.
        Since continuous functions assume their maximum over compact sets, the time function is well defined under the assumption that the machine terminates on all inputs.
        Thus for any machine $M$ there exists a function $T: \NN^\NN \times \NN$ such that for all $\varphi\in\B$ and $\str a \in \S$
        \[ \timef_M(\phi,\str a) \leq T(\length\varphi,\length{\str a}). \]
        From such a time function one may obtain a step-count just like a polynomial step-count was obtained from a polynomial running time in \Cref{resu:polytime implies opt}.
        We give a direct proof of a stronger statement by adapting Schr\"oders' methods to the \opt setting.
        \begin{lemma}[Existence of step-counts]\label{resu:existence of step-counts}
            Any oracle machine that computes a total operator has a computable step-count.
        \end{lemma}

        \begin{proof}\footnote{This proof has significanlty been simplified thanks to remarks of an anonymous referee}
            Fix an oracle machine $M$ that computes a total operator.
            Note that the time function $\timef_{M}: \B\times \Sigma^* \to\NN$, $(\varphi,\str a)\mapsto \timef_{M}(\varphi,\str a)$
            is a continuous mapping.
            The set
            \[ K_m:= \{\varphi\in\B\mid \forall n\in\NN\colon\length{\varphi}(n)\leq m \} \]
            is a compact subset of $\B$.
            Since a continuous function assumes its maximum on a compact set, we may define a function $t:\NN\to\NN$ via
            \[ t(n) := \max\{\timef_{M}(\varphi,\str a) \mid \varphi \in K_n,\length{\str a}\leq n\}. \]
            Moreover, this function is computable as its values can be obtained by simulating the computation of $M$ (again using the fact that $K_m$ is compact and therefore at some point it can be verified that the possible computation paths are exhausted).

            It remains to prove that $t$ is indeed a step count of $M$:
            Given an oracle $\varphi$ and an input $\str a$ note that since $M^\varphi(\str a)$ terminates, it only asks queries from a finite set $Q_{\varphi,\str a}$.
            As before, let $m_{\varphi,\str a}$ denote the biggest return value of $\varphi$ on any of these queries.
            Replace the oracle $\varphi$ by the oracle $\psi$ that returns the same value as $\varphi$ on any element of the queryset $Q_{\varphi,\str a}$ and the empty string otherwise.
            The runs of the machine on $\varphi$ and $\psi$ are identical, clearly $\psi \in K_{m_{\varphi,\str a}}$ and $\length{\str a} \leq m_{\varphi,\str a}$ by definition.
            Thus, $\timef_M(\varphi, \str a) = \timef_M(\psi,\str a) \leq t(m_{\varphi,\str a})$ by the definition of $t$.
            Since $\varphi$ and $\str a$ were arbitrary, this means that $t$ is a step-count for $M$.
        \end{proof}

        As an application of the above lemma and for the sake of having a rather simple construction of a counter-example oracle spelled out in this paper we revisit \Cref{ex:repeated squaring}.
        \begin{example}[Repeated squaring revisited]\label{ex:repeated squaring revisited}
            The operator $F$ from \Cref{ex:repeated squaring}, namely
            \[ F(\varphi)(\str a) := \varphi^{\length{\str a}}(\sdzero) \]
            is not polynomial-time computable.
            Thus, it can in particular not be computed by a machine that has a polynomial step-count and finite lookahead revision.

            More generally, there does not exist an oracle machine with finite lookahead revision that computes $F$.
            Towards a contradiction assume that $M$ was an oracle machine that computes $F$.
            Let $t$ be a step-count of $M$ which exists by \Cref{resu:existence of step-counts}.
            Without loss of generality assume that $t$ is strictly increasing.
            Define oracles $\varphi_n$ that force big numbers of lookahead revisions as follows.
            Let
            \[ \varphi_n(\str a) :=
                \begin{cases}
                    \sdzero^{t(n)+1} & \text{if } \str a = \sdzero \\
                    \sdzero^{t^{k+1}(n)+1} &\text{if }\str a = \sdzero^{t^k(n)+1}\text{ for some } k \\
                    \varepsilon & \text{otherwise.}
                \end{cases}
            \]
            For $n>0$ the machine $M$ encounters at least $n$ lookahead revisions in the computation $M^{\varphi_n}(\sdzero^n)$.
            To see this, first note that the query $\sdzero^{t^{n}(n) +1}$ must have been asked in this computation.
            This is because if it had not, we could change $\varphi_n$ to return $\varepsilon$ on this query instead and thereby change the return value of $F$ without giving $M$ the possibility to reconsider, which would render the value $M$ decided on incorrect.
            Next, argue by induction that for $k \leq n$ the machine $M$ can only pose queries of size at most $t^k(n)$ before doing its $k$-th length revisions.
            For $k=1$ this is due to $t$ being a step count.
            Now assume the claim has been proven for $k$.
            The time the machine is granted in this case is bounded by $t$ applied to the biggest oracle answer that has previously been given and the output.
            By the induction hypothesis and since $t$ is increasing the oracle queries that have previously been asked are bounded by $t^k(n)$, and therefore the answers are bounded by $\length{\varphi_n}(t^k(n)) \leq t^k(n)$.
              Thus, the time the machine can take without encountering another lookahead revision is bounded by $t^{k+1}(n)$.
            Thus, also the size of queries before encountering the $k+1$-st length revision is bounded by this value.

            In total it follows that the machine has to make at least $n$ length revisions to ask the query $\sdzero^{t^n(n)+1}$ which it has to ask to provide a correct return value.
            Since $t$ can be chosen computable, so can the familiy $\varphi_n$, even uniformly in the index of $M$.
            Furthermore, the familiy can be replaced by a single function using a diagonalization argument or arguing directly.
        \end{example}

\section{Lambda-calculi for feasible functionals}\label{sec:lambda}

The preceding chapter presented two classes $\spt$ and $\mpt$ of operators based on simple syntactical restrictions to POTMs, both of which fail to capture all of \p.
The rest of the paper proves that this is exclusively due to a failure of closure of these classes under composition.

Composition is a notion from the operator setting, but for this chapter the functional standpoint is more convenient.
In the functional setting, there are more ways of combining functionals.
consider for instance $F,G\colon \B\times \S\to\S$:
one may apply $G$ and hand the resulting string as input to $F$, i.e. send $\varphi$ and $\str a$ to $F(\varphi,G(\varphi,\str a))$, or leave the string argument in $G$ open and use this as function input for $F$: I.e. send $\varphi$ and $\str a$ to $F(\lambda \str b. G(\varphi,\str b),\str a)$.
The latter captures the composition of operators and uses the familiar notation for lambda-abstraction.
One may go one step further and also use lambda-abstraction over $\varphi$ and $\str a$ to express the two ways to combine functionals by terms in the lambda-calculus with $F$ and $G$ as constants.
On the other hand, any term in the lambda-calculus with constants from a given class of functionals can be interpreted as a functional again.
It should be noted that in general these functionals need not be type-one or two anymore as the lambda-calculus provides variables for each finite type.

This section reasons about closures of classes of functionals under $\lambda$-abstraction and application as a subsitute for closure under composition in the operator setting.
We do not attempt to give a self-contained presentation of the tools from lambda-calculus needed here and point to \cite{MR1241248} for more details.
The most important parts are also gathered in \ref{ap: lambda}.
Our primary focus is on using such calculi {\em definitionally} -- that is we are interested in the denotational semantics of type-one and -two terms, and only require operational notions to reduce arbitrary terms to such terms.
In particular, we consider systems with constant symbols for every function in some type-one or -two class, without necessarily giving reduction rules for such symbols.
Note that in order to take advantage of the setting of applied lambda calculus we begin by defining a class of functionals at all finite types. We then obtain subclasses of interest by the standard technique of taking \demph{sections}.

\begin{definition}\label{def:lambda definability}
For a class $\mathbf{X}$ of functionals, let $\lambda(\mathbf{X})$ denote the set of simple-typed $\lambda$-terms where a constant symbol is available for each element of $\mathbf{X}$.
For a set $T$ of terms the \demph{one-section of} $T$, denoted $\onesec{T}$, is the class of functions represented by type-one terms of $T$.
The \demph{two-section of} $T$, denoted $\twosec{T}$, is the set of functionals represented by type-two terms of $T$.
\end{definition}
Denote the class of polynomial-time computable functions by $\pf$.
It is well-known that $\onesec{\lambda(\pf)}=\pf$.
Seth proves that $\twosec{\lambda(\pf)}=C_0$, where $C_0$ is his class of functionals computed by POTMs only allowed to access their oracle a finite number of times \cite{MR1238294}.

Mehlhorn's schematic characterization of polynomial-time \cite{MR0411947} fits quite nicely into the lambda calculus approach.
However, the limited recursion on notation scheme translates to a type-three constant as it produces a type-two functional from a set of type-two functionals.
Work by Cook and Urquhart revealed that it is possible to use a type-two constant instead \cite{MR1241248}.
Cook and Urquart consider lambda-terms with symbols for a collection of basic polynomial-time computable functions as well as one type-two symbol $\Rec$ capturing limited recursion on notation.
The type-two recursion functional given in \cite{MR1241248} is slight different than $\Rec$, which is more convenient in our setting. We will use $\Rec'$ to denote the version of Cook and Urquhart. 
Note that every binary string is either the empty string $\epsilon$ or can be written as $\str c i$ for a digit $i$ and a strictly shorter string $\str c$.
Set
\begin{equation*}
\Rec'(\varphi, \str a, \psi, \epsilon) := \str a \quad\text{and}\quad
\Rec'(\varphi, \str a, \psi, \str ci) := \left\{\begin{array}{ll}
                                            \str t & \text{if $|\str t| \le |\psi(\str ci)|$;}\\
                                            \psi(\str ci) & \text{otherwise.}
                                    \end{array}\right.
\end{equation*}
where $\str t = \varphi(\str ai, \Rec'(\varphi,\str a, \psi, \str c))$.
Note that this gives the input function $\varphi$ the type $\S \times \S \to \S$ but still defines a functional of type two.
Readers familiar with the Mehlhorn's limited recursion on notation scheme should note that this is an application of the scheme to feasible functionals and that $\Rec'$ itself is a feasible functional.
The next section verifies this directly for a very similar functional.

Cook and Kapron consider a version of the Cook-Urquhart system which includes constant symbols for all type-one polynomial-time functions \cite{MR1232924}.
The classes of functionals that correspond to the terms in either of these versions coincide.
Kapron and Cook call these functionals the \demph{basic feasible functionals} and denote them by $\bff$.
Thus, $\bff=\lambda(\pf\cup\{\Rec'\})$ up to identification of lambda terms with the functionals they represent.
Adding all polynomial-time functions as symbols has the advantage that $\Rec'$ need only be used for properly higher-order recursions.

The class $\bff$ contains functionals of all finite types and can be cut down to the types we are interested in by considering the sections $\bff_1$ and $\bff_2$.
While there are reasonable doubts whether the class \bff captures feasiblity in all finite types \cite{MR3632199}, the significance of its two-section is demonstrated by  the following result of Kapron and Cook \cite{MR1374053}.
\begin{theorem}[$\p = \bff_2$]\label{thm:Kap Cook} The basic feasible functionals of type-two are exactly the polyno\-mial-time functionals.
\end{theorem}
It is also true that $\bff_1=\pf$.
Cook's notion of intuitively feasible functionals is based on this property of $\onesec{\bff}$:
Cook calls a type-two functional $F$ \demph{intuitively feasible} if the corresponding operator is in $\opt$ and adding it to \bff does not change the one-section, i.e. $\lambda(\pf\cup\{\Rec',F\})_1=\pf$.

\subsection{Limited recursion as an operator}\label{sec:limited recursion as an operator}

As mentioned, using a recursion functional slightly different from Cook and Urquart's $\Rec'$ is more convenient for our purposes.
Consider the \demph{limited recursion functional} $\Rec$, i.e. the functional defined by $\Rec(\varphi, \str a, \str b, \epsilon) := \str a$ and
\begin{align*}
\Rec(\varphi, \str a, \str b, \str{c}i) &:= \varphi(\str{c}i, \Rec(\varphi, \str a, \str b, \str{c}))^{\le |\str b|}
\end{align*}
First we establish that we may indeed swap the recursion functional $\Rec'$ with this functional.
\begin{proposition}[$\lambda(\pf\cup\{\Rec\}) = \bff$]
The recursion functionals $\Rec'$ and $\Rec$ are equivalent in the sense that
\[ \Rec \in \twosec{\lambda(\pf\cup\{\Rec'\})} \quad\text{and}\quad \Rec' \in \twosec{\lambda(\pf\cup\{\Rec\})} \]
\end{proposition}
\begin{proof}
  By definition $\Rec(\varphi,\str a, \str b, \str c) = \Rec'(\lambda \str s\lambda \str t.\varphi(\str s,\str t)^{\le |\str b|}, \str a, \lambda \str s.\str b, \str c)$.
  Since the function $\lambda \str b \lambda \str s \lambda \str t. \varphi(\str s, \str t)^{\leq \length{\str b}}$ is polynomial-time computable and can be replaced by its symbol, this proves the assertion.
  
For the other direction first show that the operator
\[ F(\psi)(\str c):= \max_{\str c'\subseteq \str c}F(\psi)(\str c') \]
may be defined using $\Rec$.
Define $F'$ such that $F(\psi)(\str c)=\psi(F'(\psi,\str c))$.
Thus,
\[
F'(\psi, \str c)=\Rec(\lambda \str s\lambda \str t.M(\psi, \str s, \str t), \epsilon, \str c, \str c),
\]
where $M(\psi, \str s, \str t)= \str s$ if $\psi(\str s) > \psi(\str t)$ and $\str t$ otherwise.
Let $\ell$ be the polynomial-time function that satisfies: $\ell(\str s,\str t)=\str s$ if $|\str s|\le|\str t|$ and $\str t$ otherwise.
Then
\[
\Rec'(\varphi,\str a, \psi, \str c) = \Rec(\lambda \str s\lambda \str t.\ell(\varphi(\str s,\str t),\psi(\str s)),\str a, F(\psi,\str c), \str c).
\]
Again by availability of symbols for $M$ and $l$, this proves the assertion.
\end{proof}

A priori, the statement that $\Rec$ is moderately polynomial-time computable does not make sense:
$\mpt$ is a class of operators.
Each of the elements $\mpt$ can be used to obtain a functional of type $\B\times \S \to \S$, but $\Rec$ does not have this type either.
To mend this, we translate $\Rec$ to an operator by relying on the tupling functions:
Recall that we denote the $k$-ary polynomial-time computable string-tupling functions by $\langle\cdot,\ldots,\cdot\rangle$ and that they have polynomial-time computable projections $\pi_{1,k},\ldots \pi_{k,k}$.
The \demph{limited recursion operator} $\mathsf R$ is defined by
\[ \mathsf R(\psi) := \lambda \str a.\Rec(\lambda\str b\lambda \str c. \psi(\langle \str b, \str c\rangle),\pi_{1,3}(\str a),\pi_{2,3}(\str a),\pi_{3,3}(\str a)). \]
The right hand side is a type-two term from $\lambda(\pf\cup\{\Rec\})$ and thus $\mathsf R \in \twosec{\lambda(\pf\cup\{\Rec\})} = \p$.
Recall that the tupling functions for string functions were defined by $\langle\varphi,\ldots,\psi\rangle(\str a)=\langle\varphi(\str a),\ldots,\psi(\str a)\rangle$.
Thus, the tupling functions on string functions are also available as type-two terms.
The same is true for the projections.

More generally, any type-two functional can be replaced by an operator using lambda-abstraction and application and additionally tupling functions and projections.
This may be regarded as a normal form of functionals of type-two that is different from the usual one that fully curries.
This provides a simple way for us to transfer from the operator setting that was used in the previous sections to the functional setting that is more appropriate for the current chapter.
We may call a type-two functional strongly, resp.\ moderately polynomial-time computable if its operator normal form is contained in $\spt$, resp.\ $\mpt$.
We refrain from studying these classes of functionals directly in this paper and only mention  without going into details that they can alternatively be described directly by addapting the notion of an oracle machine to allow for more general and multiple oracles and also addapting the access restrictions.
Instead, we concentrate on $\lambda(\spt)$ and $\lambda(\mpt)$ directly.
 
\subsection[The lambda-closures of MPT and SPT]{The lambda-closures of $\mpt$ and $\spt$}

This section proves that $\twosec{\lambda(\mpt)}=\p=\twosec{\lambda(\spt)}$.
For the first equality, the strategy is very simple:
We prove that $\mathsf R \in\mpt$ and $\lambda(\mpt)_1 = \pf$.
For the second equality additional work has to be done since $\mathsf R \notin\spt$.

\begin{lemma}[$\mathsf R \in \mpt$]\label{resu:R is mpt}
    The limited recursion operator is moderately polynomial-time computable.
\end{lemma}

\begin{proof}
    A high-level description of an oracle machine $M$ computing $\mathsf R$ can be given as follows:
    On inputs $\psi$ and $\str a$ fix the following notations: $\str t_0 := \pi_{1,3}(\str a)$, $\str c:=\pi_{2,3}(\str a)$ and $\str b:=\pi_{3,3}(\str a)$.
    The machine $M$ first queries $\psi$ at $\sdone^{\max\{\length{\langle \str c,\str t_0\rangle},\length{\langle \str c, \str b\rangle}\}}$.
    The return value is not used, but this query guarantees that $M$ has exactly one lookahead revision.
    Then, for $i=1,\dots, n=|\str c|$ set $\str t_i \leftarrow \psi(\langle \str c^{\le i}, \str t_{i-1}\rangle)^{\le |\str b|}$, and return $\str t_n$.
    Since $|\str t_i|\le|\str b|$ for $1 \le i \le n$, the initial query made by $M$ is the largest.
    Thus $M$ has a quadratic step-count and lookahead revision 1.
\end{proof}
This result does not rely on the reformulation of the recursion functional.
It can be checked that it is also true if Cook and Urquart's formulation is used.
However, the description of the machine becomes more involved.

The first main result of this section follows easily.
\begin{theorem}[$\lambda(\mpt)=\bff$]\label{resu:closure of mpt}
    The functionals represented by lambda terms with symbols for moderate polynomial-time computable operators are exactly the basic feasible functionals.
\end{theorem}

\begin{proof}
    Since $\bff = \lambda(\pf\cup\{\Rec\})$, to prove that $\lambda(\mpt)\supseteq\bff$ it suffices to specify lambda-terms in moderately polynomial-time computable operators that can be used to replace the symbols for any polynomial-time computable function and the symbol for the limited recursion functional.

    Note that for any polynomial-time computable function $\psi$ in one argument, the constant operator defined by $K_\psi(\varphi):=\psi$ is moderately polynomial-time computable.
    Thus the lambda-term $\lambda \str a. K_\psi(\lambda \str b.\str b)(\str a)$ that evaluates to the function $\psi$ may be used as replacement for a symbol for the polynomial-time computable function $\psi$.
    For multiple arguments note that the operator defined by $T(\varphi)(\str a) := \langle\varphi(\epsilon),\str a\rangle$ is moderately polynomial-time computable.
    For a $2$-ary polynomial-time computable function $\psi$ also the function $\tilde \psi(\langle\str a,\str b\rangle):= \psi(\str a,\str b)$ is polynomial-time computable and a symbol for $\psi$ may be replaced by the term
    \[ \lambda \str a\lambda \str b. K_{\tilde \psi}(\lambda \str c.\str c)(T(\lambda \str d. \str a)(\str b)). \]
    This generalizes to functions of arbitrary arity:
    For a $k$-ary $\psi$ use $$\tilde \psi(\langle \str a_1,\langle \ldots, \langle \str a_{k-1},\str a_k\rangle\ldots\rangle) := \psi(\str a_1,\ldots,\str a_k).$$

    Due to the moderate polynomial-time computability of the limited recursion operator $\mathsf R$ and the availability of tupling functions from the first part of the proof, the lambda-term
    \[ \lambda \varphi\lambda\str a\lambda \str b\lambda \str c.\mathsf{R}(\lambda\str d.\varphi(\pi_1(\str d),\pi_2(\str d)))(\langle \str a, \str b,\str c \rangle) \]
    may be used to replace the symbol $\Rec$.

    That the lambda-closure does not lead outside of $\bff$ follows from the inclusion of $\mpt$ in $\p$ that was provided in Proposition~\ref{resu:from step-count and lookahead revision number to running time} together with $\lambda(\p)\subseteq\bff$ which follows from Kapron and Cook's theorem that $\p = \bff_2$ (i.e.\ Theorem~\ref{thm:Kap Cook}).
\end{proof}

Unfortunately the same tactic  is bound to fail for the strongly polynomial-time computable operators:
an argument similar to that given for the maximization operator in Example~\ref{ex:not spt} shows that the limited recursion operator is not in \spt.
This forces us to attempt to split $\Rec$ into simpler parts.
Due to the concrete form of our limited recursion functional, this can fairly easily be done.
It should be noted though, that this is a special case of a more general theorem proved in the next section and that understanding the decomposition in the next lemma is not crucial for the understanding of the paper.

\begin{lemma}[$\Rec \in \twosec{\lambda(\spt)}$]\label{resu:R in sptcircspt}
   There exists a lambda-term with constants from $\spt$ that evaluates to the limited recursion functional $\Rec$.
\end{lemma}

\begin{proof}
For convenience we use oracle machines with multiple oracles and the corresponding generalisation of $\spt$.
Functionals $\mathcal S,\mathcal T\in\spt$ such that for all
$\varphi,\str a, \str b, \str c$,
\begin{equation}\label{eq:R decomp}
\Rec(\varphi, \str a, \str b, \str c)=\mathcal{T}(\lambda \str d \lambda \str t.\mathcal{S}(\varphi, \str t, \str b, \str d), \str a, \str b, \str c)
\end{equation}
can be obtained by considering oracle machines $M$ and $N$ that work as follows:
$N$ computes the function $\mathcal{S}(\varphi, \str t, \str b, \str d) = \varphi(\str d, \str t)^{\le |\str b|}$ in a straightforward way.
The following is a high-level description of an oracle machine $M$ computing $\mathcal{T}$:
on inputs $\psi, \str a, \str b, \str c$, letting $\str t_0 = \str a$, for $i=1,\dots, n=|\str c|$, it sets $\str s \leftarrow \psi(\str c^{\le i}, \str t_{i-1})$, and if $|\str s| \le |\str b|$ set $\str t_i \leftarrow \str s$, otherwise halt and return $\epsilon$.
If all $n$ steps complete, it returns $\str t_n$.
Thus $M$ computes $\mathcal{T}$ has a quadratic step-count and length revision 1.
Moreover, $\mathcal S$ and $\mathcal T$ satisfy \eqref{eq:R decomp}.
\end{proof}
The following theorem can now be proven completely analogous to the same statement for $\mpt$.

\begin{theorem}[$\lambda(\spt) = \bff$]\label{resu:closure of spt}
    The functionals represented by lambda terms with symbols for strongly polynomial-time computable operators are exactly the basic feasible functionals.
\end{theorem}

\begin{proof}
    The only adjustment that has to be made is to change the lambda term replacing $\Rec$ to the term over $\spt$ that exists according to Lemma~\ref{resu:R in sptcircspt}.
\end{proof}

As we are mainly interested in the two section $\p$ of the basic feasible functionals we gather these special cases from the Theorems~\ref{resu:closure of mpt} and Theorem~\ref{resu:closure of spt} and state them as the main results of this paper.

\begin{theorem}[$\twosec{\lambda(\mpt)}=\p=\twosec{\lambda(\spt)}$]\label{th: main}
The closure of $\spt$ and $\mpt$ under lambda-abstraction and application are exactly the polynomial-time computable functionals.
\end{theorem}

\subsection{Some results about composition}\label{sec:some results about composition}

The decomposition of the limited recursion functional into a lambda-term over two strongly polynomial-time computable functionals from Lemma~\ref{resu:R in sptcircspt} is a special case of a general phenomenon.
For the statement of the corresponding result we switch back to the operator setting and provide a decomposition of any moderately polynomial-time computable operator into a composition of two strongly polynomial-time computable operators.
Note that composition of operators may be expressed as a lambda term: $S\circ T = \lambda\varphi.(S\circ T)(\varphi) = \lambda\varphi.S(T(\varphi))$.
Thus, Lemma~\ref{resu:R in sptcircspt} can be recovered by decomposing the limited recursion operator and then using that the limited recursion functional can be expressed as a lambda term over the limited recursion operator.

It is instructive to revisit the maximization operator from Example~\ref{ex:not spt} with this in mind.
\begin{example}[Revisiting Example~\ref{ex:not spt}]\label{ex:decompose}
  Recall that Example~\ref{ex:not spt} concluded that the operator $F(\varphi)(\str a) := \sdone^{\max_{\str b\subseteq\str a}\length{\varphi(\str a)}}$ is not contained in $\spt$.
  It is fairly straightforward to write this operator as a composition of two operators from $\spt$:
  Let $G(\varphi)$ be the function that interprets its string input as a finite list of strings $\str a_1,\ldots,\str a_n$ and returns $\sdzero^{n-k}\sdone^k$ if $k< n$ is the least index such that $\length{\varphi(\str a_k)} > \length{\varphi(\str a_{k+1})}$ and $\sdone^{\max_{k < n}\{\length{\varphi(\str a_k)}\}}$ otherwise.
  It is straightforward to implement $G$ by an oracle machine that proves it to be contained in $\spt$.
  Next let $M$ be the machine that on oracle $\varphi$ and input $\str a$ acts as follows:
  It first produces the list $L$ of all initial segments of $\str a$ and queries its oracle on this list.
  If the return value is neither of the form $\sdzero^{\length{\str a} -k}\sdone^k$ for some $0<k \leq \length{\str a}$ nor $\sdone^m$ it returns $\epsilon$.
  If the return value is of the form $\sdzero^{\length{\str a} -k}\sdone^k$ it swaps the $k$-th and the $k+1$-st entry of the list and queries the oracle on the new list.
  If this is the case more than $\length{\str a}^2$ times the machine terminates and returns $\epsilon$.
  If the return value is $\sdone^m$ it returns $\sdone^m$.
  From this definition is is directly apparent that $M$ always terminates and computes some total operator $H$ that $M$ proves to be from $\spt$.
  It should be clear that $H \circ G = F$: $G$ checks if the initial segments are in an order such that the return values are increasing and provides $H$ with the necessary information for reordering in the case that this is not true.
  
  Alternatively one can stay closer to Lemma~\ref{resu:R in sptcircspt} and note that $F(\varphi)(\str a)=\length{\varphi(F'(\varphi,\str a))}$,
where $F'(\varphi,\str a)=\argmax_{\str a' \subseteq \str a}\varphi(\str a')$.
This $F'$ may be defined as
\[
F'(\varphi, \str a) = \mathcal{T}(\lambda \str d\lambda \str t.M(\varphi,\str d,\str t),\epsilon,\str a,\str a)
\]
where $\mathcal{T} \in \spt$ is defined as in the proof of Lemma~\ref{resu:R in sptcircspt} and
\[
M(\psi, \str s, \str t)=\left\{\begin{array}{ll}
                            \str s & \text{if $\psi(\str s) > \psi(\str t)$;}\\
                            \str t & \text{otherwise.}
                        \end{array}\right.
\]
Since $M$ makes only two queries it is immediately seen to be in $\spt$.
\end{example}

The decomposition from Lemma~\ref{resu:R in sptcircspt} is a special case of the following theorem.

        \begin{theorem}[$\mpt\subseteq \spt\circ\spt$]\label{resu:factorization}
            Any moderate polynom\-ial-time computable operator can be written as composition of two strongy polynomial-time computable operators.
        \end{theorem}

        \begin{proof}
            Let $M$ be a machine that proves that an operator is from $\mpt$.
            The core ideas behind the decomposition are all already present in case that was discussed in Example~\ref{ex:decompose}: First modify $M$ to a machine $\tilde M$ that first does some preparation, then carries out the computation that $M$ does but aborts with an error message whenever it runs into a proper length revision.
            Here, \demph{proper} means that the length revision does not happen on the first oracle query.
            Of course, this machine does not compute the same operator as $M$ as the computation may be aborted.
            To fix this precompose another machine $N$ that catches the error message, and calls $\tilde M$ again, providing it with information from the error message that makes it possible for the preparation process to avoid the simulation from running into the same problem again.

            To give a detailed description of what the factors $\tilde M$ and $N$ do, we use an additional separator symbol $\#$.
            The availability of such a symbol can easily be simulated with any alphabet containing at least two symblos $\sdzero$ and $\sdone$ by replacing each occurence of $\sdzero$ by $\sdzero\sdzero$, each occurence of $\sdone$ by $\sdzero\sdone$ and $\#$ with $\sdone\sdone$.

            The machine $\tilde M$ works as follows:
            With oracle $\varphi$ and on input $\langle \str a, \str b\rangle$ it first checks if the string $\str b$ is of the form $\str c \# \str c'$.
            If it is not it returns the empty string.
            If it is, then it asks the query $\str c$.
            The answer to that first query is ignored and instead the steps $M$ does on input $\str a$ and with oracle $\varphi$ are carried out while keeping track of the size $m$ of the biggest oracle answer.
            In case a length revision is encountered on oracle query $\str d$, the machine checks if $\str d$ is shorter than the string $\str c'$.
            If so, $\tilde M$ returns $\str d \# \sdone^{\length{\str c'} - \length {\str d}}$.
            If not, it returns $\str d \#\# \sdone^m$.

            To describe how the machine $N$ works, let $p$ be a polynomial step-count of the original machine and $r$ its number of lookahead revisions.
            With oracle $\psi$ and on input $\str a$ the machine $N$ behaves as follows:
            It evaluates $m:=p(\length{\str a})$ and poses the query $\langle \str a,\#\sdone^m\rangle$ to the oracle.
            If the answer contains a $\#$ and triggers a length revision, it returns the empty string.
            If the answer contains a single $\#$ and does not trigger a length revision, it copies the string $\str d$ of digits that occur before the first $\#$ and poses as next query $\langle \str a,\str d\#\sdone^m\rangle$.
            If the answer contains a double $\#$, it checks whether the length $k$ of the string after the double $\#$ is bigger than $\length{\str a}$ and if so replaces $m$ by $p(k)$.
            If it runs into the $2r+1$-st length revision, the machine $N$ returns the empty string.

            By definition the machines $\tilde M$ and $N$ both have finite length revision.
            It can be checked that they also have polynomial step-counts.
            To see that the composition of the two machines computes the same operator as $M$, it is sufficient to argue that the only time a length revision may happen for the machine $N$ is on an oracle query where the execution of the machine $\tilde M$ runs across a lookahead revision in $M$ that it has not come across before and on the query directly following such a query.
            Indeed: If a length revision happens for $N$ this means that either $\tilde M$ has run into the situation that $M$ attempted to ask an oracle query that was bigger than anything that could have been asked at the point where the last lookahead revision has happened (i.e.\ the last time a length revision happened for $N$).
            Or the machine $N$ has just updated the standard size for the return values $\tilde M$ is supposed to give.
        \end{proof}
To illustrate how the general decomposition proceeds, we informally describe its effect on the limited recursion functional $\Rec$ from Section~\ref{sec:limited recursion as an operator}.
In this case we have
\[
        \Rec(\varphi,\str a, \str b, \str c)=\mathcal{P}(\lambda \str d\lambda \str t.\mathcal{Q}(\varphi, \str t, \str b, \str c, \str d), \str a, \str b, \str c),
\]
where $\mathcal Q$ acts like a re-entrant version of $\mathcal{R}$ which may be started at an arbitrary point in the recursion, and raises an exception whenever it encounters more than one length revision.
$\mathcal{P}$ acts as an exception handler that restarts $\mathcal Q$ in case an exception is thrown.
The following is a high-level description of an oracle machine $M$ computing $\mathcal Q$.
On inputs $\varphi, \str t, \str b, \str c, \str d$, first check that $\str d \subseteq \str c$.
If not, return an ``abort'' value, say $\epsilon$.
Otherwise, letting
$\str t_0 = \str t$,
$M$ does the following for $i=1,\dots, n=|\str c|-|\str d|$:
\begin{enumerate}
\item Set $\str c_i \leftarrow \str c^{\le (i+|\str d|)}$, $\str s_i \leftarrow \varphi(\str c_i, \str t_{i-1})$ and  $\str t_i \leftarrow \str s_i^{\le |\str b|}$
\item If $i > 1$ and $|\str s_i| > |\str s_{i-1}|$ return $\str c_i$ and $\str t_i$ encoded as one string whose length is only dependent on $\length{\str b}$.
\end{enumerate}
If all steps execute, return $\str t_n$ marked as return value.
In this case it is simple to guarantee that all exception messages have the same length and that this makes it possible for the exception handler $\mathcal P$ to avoid length revisions.
Now an OTM $N$ computing, $\mathcal P$, on inputs $\psi, \str a, \str b, \str c$, repeatedly calls $\psi$, starting with inputs $\str \epsilon, \str a$.
If any answer is marked as return value then $N$ returns this answer.
If the answer is an exception message, $N$ decodes the exception and feeds it back to $\psi$.
If none of this, more than two length revisions occur or more than $\length{\str c}$ iterations pass, then $N$ aborts and returns $\epsilon$.


Methods similar to those employed in the proof of Theorem~\ref{resu:factorization} can be used to prove that in special cases composition does not lead outside of $\mpt$.

        \begin{lemma}\label{resu:composition}
            Let $F, G\in \mpt$ be such that there exists a machine proving $G\in \mpt$ that does only one lookahead revision.
            Then $F\circ G\in\mpt$.
        \end{lemma}
        \begin{proof}
            Let $M$ and $N$ be oracle machines that run with finite lookahead revisions and have polynomial step-counts $p$ resp.\ $q$ and such that $N$ only makes one lookahead revision.
            Modify the machines $M$ and $N$ to machines $M'$ and $N'$ as follows:
            When given an input $\str a$, $M'$ first computes $p(\length{\str a})$ and then carries out the operations that $M$ carries out while checking for each oracle query $\str b$ that $M$ issues whether the length is bounded by $p(\length{\str a})$.
            If this is the case it modifies the query to be $\langle \str b,p(\length{\str a})\rangle$ and continues with the computation of $M$ while remembering the size of the biggest oracle answer that turns up in this computation.
            Otherwise it computes $p(m)$ where $m$ is the maximum of $\length{\str a}$ and the  size of the biggest oracle answer it has witnessed so far and continues with the above procedure where $p(\length{\str a})$ is replaced by $p(m)$.
            Since $p$ is a step-count for $M$, it must be the case that the length of the query is bounded by $p(m)$, so that the simulation may proceed at this point.
            Once the machine $M$ terminates, $M'$ returns its return value.

            The machine $N'$ on the other hand works as follows:
            It interprets the input as a pair $\langle \str b,m\rangle$.
            If this is not possible, it aborts and returns the empty string.
            If it is possible, it first computes $q(m)$ and issues the query $\sdone^{q(m)}$ and then carries out the steps that $N$ would have carried out on input $\str b$.
            Once $N$ terminates, it hands the return value on.

            It can be checked that the machine $M'^{N'}$ that arises by replacing each oracle call in $M'$ with a call to a subroutine that works like $N'$ computes the same function as $M^N$ does.
            To see that $M'^{N'}$ runs with finite lookahead revision first note that it is guaranteed that all the $\str b$ that are handed to $N'$ are in length smaller than the value of the second argument.
            Thus, $N'$ will not do a new lookahead revision unless the second value of the pairs it is given changes.
            That the second value only changes a finite number of times is guaranteed by the finite lookahead revision of $M$.

            That the machine $M'^{N'}$ has a polynomial step-count follows by checking that each of $M'$ and $N'$ have polynomial step counts.
            Thus the composition of the machines runs in polynomial time and thereby has a polynomial step-count.
        \end{proof}
The limited-recursion operator and $C_0$ are included in the restricted class of operators from the previous lemma.

The proofs of \Cref{resu:factorization} and \Cref{resu:composition} proceeded by modifying given machines.
Indeed, the closure property from \Cref{resu:composition} breaks down on a machine level:
\begin{example}[Composition of machines]
    Consider the operator
    \[ F:\B\to\B,\quad F(\varphi)(\str a):= \sdzero^{\max\{\length{\varphi(\sdzero^n)}\mid n\leq \length{\str a}\}}. \]
    This operator can be computed by the machine $M$ that acts as follows when given an oracle $\varphi$ and input $\str a$:
    It queries the oracle $\length{\str a}$ times, namely at $\varphi(\sdzero^i)$ where $i$ decreases from $\length{\str a}$ to zero.
    Each time after issuing a query it saves the maximum of the length of the return value and the current content of the first memory tape to the first memory tape.
    Obviously this machine has a polynomial step-count and only does one lookahead revision.

    Define another operator $G$ by
    \[ G(\varphi)(\str a):= \varphi(\varphi(\str a)). \]
    This operator can be computed in the straightforward way by a machine $N$ with lookahead revision two.

    The composition $F\circ G$ is computed by the machine $M^N$ that can be obtained by replacing each oracle call of $M$ with a subroutine that imitates what $N$ does.
    However, this machine is easily seen to not have finite lookahead revision:
    For any $N$, executing the machine $M^N$ on input $\sdzero^k$ and any oracle such that for all $n\leq k$ it holds that $\varphi(\sdzero^n) = \sdzero^{2k-n}$ results in $N$ lookahead revisions.
\end{example}
The standard proof that polynomial-time computable operators are closed under composition (i.e.\ the proof of \Cref{resu: closure under composition}), by contrast, works on a machine level:
To prove polynomial-time computability of the composition of two polynomial-time computable operators, a machine is constructed from fast machines for the two operators in the straightforward way, i.e. by replacing the oracle calls in the program of the first machine by copies of the program of the second machine.
Finally note that the proofs that $\mpt$ and $\spt$ are included in $\p$ proceed by proving the very machines that witness that an operator is contained in $\mpt$ or $\spt$ to run in polynoimal time.
Thus, feasibility of an operator can still be proven by decomposing it into moderately or strongly polynomial-time computable components and combining the corresponding machines in the straightforward way.

\section{Conclusions and future work}
We have given two characterisations of feasible type-two computation which coincide with the familiar notion of polynomial-time, but have a simple and appealing syntactic description.
They use POTMs with simple restrictions on how oracles may be accessed as building blocks.
Such machines may call other POTMs as subroutines, but as long as all the machines obey the query restriction, the result is polynomial-time.
Although we do not consider this the main contribution -- the evidence is overwhelming already -- this further supports the naturalness of the notion of feasibility in second-order complexity theory:
both of these models, formulated without any notion of length-functions, second-order polynomials, or limited recursion on notation, lead to the familiar notion.
The simplicity of the characterisation should make it easier to reason about feasibility in the type-two setting.

While the results of this paper are satisfactory, they do raise a lot of additional questions.
For instance we conclude that $\spt\subsetneq \mpt\subseteq\spt\circ\spt \subseteq \p$ and that at least one more inclusion must be strict (as $\mpt\subsetneq \p$).
We believe both of these and many similar inclusions combining more operators from $\spt$ or $\mpt$ to be strict.
We tried to prove the equality of $\mpt\circ\mpt$ and $\p$ early in our search for closure properties of $\mpt$.
This lead us to ideas very similar to those pursued by Seth \cite{MR1238294}, who also ends up not following the straightforward path and instead taking a detour through lambda calculus.
We do fail to produce a counterexample, though.
For instance the functional from Example~\ref{ex: filr not in p} can be written as $F(\varphi,\str a)=\Rec(\lambda \str b. \varphi(\varphi(\str b)), \epsilon, \varphi(\epsilon), \str a)$.
Both the limited recursor $\Rec$ and $\lambda \str b.\varphi(\varphi(\str b))$ are from $\mpt$.

This leads to another goal, namely a more direct proof of the closure results that may provide a concrete decomposition into few elements of \spt or \mpt.
The number of components needed may provide a measure for the complexity of a task that resolves finer than polynomial-time.
This is in particular interesting as there does not exist a second-order substitute for the degree of a polynomial or even linearity.

Finally, the results of this paper lead naturally to questions of related characterizations. One such question is whether notions of length- and lookahead-revision may be applied in a linguistic setting, with either a scheme-based or applied lambda-calculus approach.
Recently, Kapron and Steinberg have answered this question in the affirmative by giving a characterziation of second-order polynomial time in an applied lambda-calculus with type-two functionals that capture \demph{length-iteration with bounded oracle revision} \cite{KS19}.
Another related question is whether it is possible to give a characterization using an implicit complexity approach, either in the style of Bellantoni and Cook \cite{MR1190824} or Leivant \cite{Leivant93}.
A positive answer has been given in this case as well.
Hainry et. al. recently presented an imperative programming language for type-two polynomial time computation using a tiered type system that controls the flow of data and also enforces finite lookahead revision \cite{HKMP}.
These results build on previous work of Marion \cite{MR2858884}.

\paragraph*{Acknowledgements}
The second author would like to thank Eike Neumann and Matthias Schr\"oder for extended discussion on several topics related to the content of this paper.

\bibliographystyle{plainurl}
\bibliography{LICSbib}

\begin{thebibliography}{10}

\bibitem{BeameCEIP98}
Paul Beame, Stephen~A. Cook, Jeff Edmonds, Russell Impagliazzo, and Toniann
  Pitassi.
\newblock The relative complexity of {NP} search problems.
\newblock {\em J. Comp. Sys. Sci.}, 57(1):3--19, 1998.
\newblock \href {https://doi.org/10.1006/jcss.1998.1575}
  {\path{doi:10.1006/jcss.1998.1575}}.

\bibitem{MR1190824}
Stephen Bellantoni and Stephen Cook.
\newblock A new recursion-theoretic characterization of the polytime functions.
\newblock {\em Comput. Complexity}, 2(2):97--110, 1992.
\newblock URL:
  \url{https://doi-org.ezproxy.library.uvic.ca/10.1007/BF01201998}, \href
  {https://doi.org/10.1007/BF01201998} {\path{doi:10.1007/BF01201998}}.

\bibitem{MR973445}
Jonathan~F. Buss.
\newblock Relativized alternation and space-bounded computation.
\newblock {\em J. Comp. Sys. Sci.}, 36(3):351--378, 1988.
\newblock Structure in Complexity Theory Conference (Berkeley, CA, 1986).
\newblock \href {https://doi.org/10.1016/0022-0000(88)90034-7}
  {\path{doi:10.1016/0022-0000(88)90034-7}}.

\bibitem{MR1911553}
Samuel~R. Buss and Bruce~M. Kapron.
\newblock Resource-bounded continuity and sequentiality for type-two
  functionals.
\newblock {\em ACM Trans. Comp. Log.}, 3(3):402--417, 2002.
\newblock \href {https://doi.org/10.1145/507382.507387}
  {\path{doi:10.1145/507382.507387}}.

\bibitem{Cobham65}
Alan Cobham.
\newblock The intrinsic computational difficulty of functions.
\newblock In Yehoshua Bar{-}Hillel, editor, {\em Logic, Methodology and
  Philosophy of Science: Proc. 1964 Intl. Congress (Studies in Logic and the
  Foundations of Mathematics)}, pages 24--30. North-Holland Publishing, 1965.

\bibitem{Constable73}
Robert~L. Constable.
\newblock Type two computational complexity.
\newblock In {\em 5th Annual {ACM} STOC, (Austin, {TX} 1973)}, pages 108--121,
  1973.
\newblock \href {https://doi.org/10.1145/800125.804041}
  {\path{doi:10.1145/800125.804041}}.

\bibitem{MR1241248}
Stephen Cook and Alasdair Urquhart.
\newblock Functional interpretations of feasibly constructive arithmetic.
\newblock {\em Ann. Pure Appl. Logic}, 63(2):103--200, 1993.
\newblock \href {https://doi.org/10.1016/0168-0072(93)90044-E}
  {\path{doi:10.1016/0168-0072(93)90044-E}}.

\bibitem{Cook71}
Stephen~A. Cook.
\newblock The complexity of theorem-proving procedures.
\newblock In {\em 3rd Annual {ACM} STOC, 1971, Shaker Heights, {OH}}, pages
  151--158, 1971.
\newblock \href {https://doi.org/10.1145/800157.805047}
  {\path{doi:10.1145/800157.805047}}.

\bibitem{MR1236005}
Stephen~A. Cook.
\newblock Computability and complexity of higher type functions.
\newblock In {\em Logic from computer science ({B}erkeley, {CA}, 1989)},
  volume~21 of {\em Math. Sci. Res. Inst. Publ.}, pages 51--72. Springer, New
  York, 1992.
\newblock \href {https://doi.org/10.1007/978-1-4612-2822-6_3}
  {\path{doi:10.1007/978-1-4612-2822-6_3}}.

\bibitem{MR1232924}
Stephen~A. Cook and Bruce~M. Kapron.
\newblock Characterizations of the basic feasible functionals of finite type.
\newblock In {\em Feasible mathematics ({I}thaca, {NY}, 1989)}, volume~9 of
  {\em Progr. Comput. Sci. Appl. Logic}, pages 71--96. Birkh\"auser, 1990.

\bibitem{MR2295797}
Norman Danner and James~S. Royer.
\newblock Adventures in time and space.
\newblock {\em Log. Methods Comput. Sci.}, 3(1):1:9, 53, 2007.
\newblock \href {https://doi.org/10.2168/LMCS-3(1:9)2007}
  {\path{doi:10.2168/LMCS-3(1:9)2007}}.

\bibitem{MR3632199}
Hugo F\'er\'ee.
\newblock Game semantics approach to higher-order complexity.
\newblock {\em J. Comput. System Sci.}, 87:1--15, 2017.
\newblock \href {https://doi.org/10.1016/j.jcss.2017.02.003}
  {\path{doi:10.1016/j.jcss.2017.02.003}}.

\bibitem{HKMP}
Emmanuel Hainry, Bruce~M. Kapron, Jean-Yves Marion, and Romain P{\'e}choux.
\newblock A tier-based characterization of the basic feasible functionals.
\newblock {\em In preparation}.

\bibitem{MR2053401}
Aleksandar Ignjatovic and Arun Sharma.
\newblock Some applications of logic to feasibility in higher types.
\newblock {\em ACM Trans. Comput. Log.}, 5(2):332--350, 2004.

\bibitem{MR1826285}
Robert~J. Irwin, James~S. Royer, and Bruce~M. Kapron.
\newblock On characterizations of the basic feasible functionals. {I}.
\newblock {\em J. Funct. Programming}, 11(1):117--153, 2001.
\newblock \href {https://doi.org/10.1017/S0956796800003841}
  {\path{doi:10.1017/S0956796800003841}}.

\bibitem{MR1374053}
Bruce~M. Kapron and Stephen~A. Cook.
\newblock A new characterization of type-{$2$} feasibility.
\newblock {\em SIAM J. Comput.}, 25(1):117--132, 1996.
\newblock \href {https://doi.org/10.1137/S0097539794263452}
  {\path{doi:10.1137/S0097539794263452}}.

\bibitem{KS19}
Bruce~M. Kapron and Florian Steinberg.
\newblock Type-two iteration with bounded query revision.
\newblock {\em In preparation}.

\bibitem{KawamuraC12}
Akitoshi Kawamura and Stephen~A. Cook.
\newblock Complexity theory for operators in analysis.
\newblock {\em {TOCT}}, 4(2):5:1--5:24, 2012.
\newblock \href {https://doi.org/10.1145/2189778.2189780}
  {\path{doi:10.1145/2189778.2189780}}.

\bibitem{MR3259646}
Akitoshi Kawamura and Hiroyuki Ota.
\newblock Small complexity classes for computable analysis.
\newblock In {\em Mathematical foundations of computer science 2014. {P}art
  {II}}, volume 8635 of {\em LNCS}, pages 432--444. Springer, Heidelberg, 2014.
\newblock \href {https://doi.org/10.1007/978-3-662-44465-8_37}
  {\path{doi:10.1007/978-3-662-44465-8_37}}.

\bibitem{KawamuraS17}
Akitoshi Kawamura and Florian Steinberg.
\newblock Polynomial running times for polynomial-time oracle machines.
\newblock In {\em 2nd International Conference on Formal Structures for
  Computation and Deduction, 2017, Oxford, {UK}}, pages 23:1--23:18, 2017.
\newblock \href {https://doi.org/10.4230/LIPIcs.FSCD.2017.23}
  {\path{doi:10.4230/LIPIcs.FSCD.2017.23}}.

\bibitem{Leivant93}
Daniel Leivant.
\newblock Stratified functional programs and computational complexity.
\newblock In Mary S.~Van Deusen and Bernard Lang, editors, {\em 20th Annual
  {ACM} Symposium on Principles of Programming Languages (POPL)}, pages
  325--333. {ACM} Press, 1993.

\bibitem{MR2048061}
Daniel Leivant.
\newblock Implicit computational complexity for higher type functionals.
\newblock In {\em Computer science logic}, volume 2471 of {\em LNCS}, pages
  367--381. Springer, Berlin, 2002.

\bibitem{MR2858884}
Jean-Yves Marion.
\newblock A type system for complexity flow analysis.
\newblock In {\em 26th {A}nnual {IEEE} {S}ymposium on {L}ogic in {C}omputer
  {S}cience---{LICS} 2011}, pages 123--132. IEEE Computer Soc., Los Alamitos,
  CA, 2011.

\bibitem{MR0411947}
Kurt Mehlhorn.
\newblock Polynomial and abstract subrecursive classes.
\newblock {\em J. Comp. Sys. Sci.}, 12(2):147--178, 1976.
\newblock \href {https://doi.org/10.1016/S0022-0000(76)80035-9}
  {\path{doi:10.1016/S0022-0000(76)80035-9}}.

\bibitem{MR1251285}
Christos~H. Papadimitriou.
\newblock {\em Computational complexity}.
\newblock Addison-Wesley Publishing Company, Reading, MA, 1994.

\bibitem{MR1727821}
Elena Pezzoli.
\newblock On the computational complexity of type 2 functionals.
\newblock In {\em Computer science logic ({A}arhus, 1997)}, volume 1414 of {\em
  LNCS}, pages 373--388. Springer, Berlin, 1998.
\newblock \href {https://doi.org/10.1007/BFb0028026}
  {\path{doi:10.1007/BFb0028026}}.

\bibitem{MR1463765}
James~S. Royer.
\newblock Semantics vs syntax vs computations: machine models for type-2
  polynomial-time bounded functionals.
\newblock {\em J. Comp. Sys. Sci.}, 54(3):424--436, 1997.
\newblock \href {https://doi.org/10.1006/jcss.1997.1487}
  {\path{doi:10.1006/jcss.1997.1487}}.

\bibitem{MR2075336}
James~S. Royer.
\newblock On the computational complexity of {L}ongley's {$H$} functional.
\newblock {\em Theoret. Comp. Sci.}, 318(1-2):225--241, 2004.
\newblock \href {https://doi.org/10.1016/j.tcs.2003.10.024}
  {\path{doi:10.1016/j.tcs.2003.10.024}}.

\bibitem{MR2090390}
Matthias Schr\"oder.
\newblock Spaces allowing type-2 complexity theory revisited.
\newblock {\em MLQ Math. Log. Q.}, 50(4-5):443--459, 2004.
\newblock \href {https://doi.org/10.1002/malq.200310111}
  {\path{doi:10.1002/malq.200310111}}.

\bibitem{SethRec}
Anil Seth.
\newblock There is no recursive axiomatization for feasible functionals of type
  {$2$}.
\newblock In {\em 7th {A}nnual {IEEE} {S}ymposium on {L}ogic in {C}omputer
  {S}cience ({S}anta {C}ruz, {CA}, 1992)}, pages 286--295. IEEE Comput. Soc.
  Press, Los Alamitos, CA, 1992.
\newblock \href {https://doi.org/10.1109/LICS.1992.185541}
  {\path{doi:10.1109/LICS.1992.185541}}.

\bibitem{MR1238294}
Anil Seth.
\newblock Some desirable conditions for feasible functionals of type {$2$}.
\newblock In {\em 8th {A}nnual {IEEE} {S}ymposium on {L}ogic in {C}omputer
  {S}cience ({M}ontreal, {PQ}, 1993)}, pages 320--331. IEEE Comput. Soc. Press,
  Los Alamitos, CA, 1993.
\newblock \href {https://doi.org/10.1109/LICS.1993.287576}
  {\path{doi:10.1109/LICS.1993.287576}}.

\bibitem{MR2103646}
Thomas Strahm.
\newblock A proof-theoretic characterization of the basic feasible functionals.
\newblock {\em Theoret. Comput. Sci.}, 329(1-3):159--176, 2004.

\bibitem{MR1058425}
Mike Townsend.
\newblock Complexity for type-{$2$} relations.
\newblock {\em Notre Dame J. Formal Logic}, 31(2):241--262, 1990.
\newblock \href {https://doi.org/10.1305/ndjfl/1093635419}
  {\path{doi:10.1305/ndjfl/1093635419}}.

\end{thebibliography}
\newpage

\appendix
\hidefromtoc
\section[Lambda-Definability]{$\lambda$-Definability}\label{ap: lambda}

The treatment of the typed $\lambda$-calculus here follows that in
\cite{MR1241248}.

The set of {\em types} is defined inductively as follows:
\begin{itemize}
\item
0 is a type
\item
$(\sigma \rightarrow \tau)$ is a type, if $\sigma$ and $\tau$ are types.
\end{itemize}

The set $Fn(\tau)$ of  \demph{functionals of type} $\tau$ is defined
by induction on $\tau$:
\begin{itemize}
\item
$Fn(0) = \S$
\item
$Fn(\sigma \rightarrow \tau) = \{F | F: Fn(\sigma) \rightarrow Fn(\tau)\}$.
\end{itemize}
It is not hard to show that each type $\tau$ has a unique normal form
\[
\tau = \tau_1 \rightarrow \tau_2 \rightarrow \cdots \rightarrow \tau_k \rightarrow 0
\]
where the missing parentheses are put in with association to the right.
Hence a functional $F$ of type $\tau$ is considered in a natural
way as a function of variables $X_1, \ldots , X_k$, with $X_i$
ranging over $Fn(\tau_i)$, and returning a natural number value:
$$F(X_1)(X_2) \ldots (X_k) = F(X_1, \ldots , X_k).$$

The \demph{level} of a type is defined inductively: The
level of type 0 is 0, and the level of the type $\tau$ written in the above normal form
is 1 + the maximum of the levels of $\tau_1, \ldots, \tau_k$.

Let $\mathbf{X}$ be a collection of type-2 functionals, that is a collection such that each functional may be assigned a type $\sigma$ with level 2.
The set of $\lambda$-\demph{terms} over $\mathbf{X}$, denoted $\lambda(\mathbf{X})$ is
defined as follows:
\begin{itemize}
\item
For each type $\sigma$ there are infinitely many
variables $X^\sigma, Y^\sigma, Z^\sigma, \ldots$ of type $\sigma$, and each
such variable is a term of type $\sigma$.
\item
For each functional $F$ (of type $\sigma$) in $\mathbf{X}$ there is a term $F^\sigma$ of type $\sigma$.
\item
If $T$ is a term of type $\tau$ and $X$ is a variable
of type $\sigma$, then $(\lambda X.T)$ is a term of type
$(\sigma \rightarrow \tau)$ (an abstraction).
\item
If $S$ is a term of type $(\sigma \rightarrow \tau)$ and $T$ is a term
of type $\sigma$, then $(ST)$ is a term of type $\tau$
(an application).
\end{itemize}

For readability, we write $S(T)$ for $(ST)$; we also write $S(T_1, \ldots , T_k)$ for $(\ldots ((ST_1)T_2) \ldots T_k)$, and $\lambda X_1 \ldots \lambda X_k.T$ for $(\lambda X_1.(\lambda X_2.( \ldots (\lambda X_k.T) \ldots )))$.

The set of free variables of a lambda term can be defined inductively and are those that are not bound by a lambda abstraction.
A term is called closed, if it has no free variables.
In a natural way each closed $\lambda$-term $T$ of type $\tau$ represents a functional in $Fn(\tau)$.
This correspondence is demonstrated in the standard way, by showing that a mapping of variables to functionals with corresponding type can be extended to a mapping of terms to functionals with corresponding type.

An \demph{assignment} is a mapping $\phi$ taking variables to functionals with corresponding type.
Suppose $\phi$ is an assignment and $T$ a $\lambda$-term over $\mathbf X$.
The value $\mathcal{V}_\phi(T)$ of $T$ with respect to $\phi$ is defined by induction on $T$ as follows.

When $T$ is a variable, $\mathcal{V}_\phi(T)$ is just $\phi(T)$.
If $T = F^\sigma$ is a constant symbol for some $F\in \mathbf X$, then $\mathcal{V}_\phi(T) = F$.

Suppose that $\tau = \tau_1 \rightarrow \ldots \rightarrow \tau_k \rightarrow 0$.
When $T$ has the form $\lambda X^\sigma. S^\tau$, $F$ is a type $\sigma$ functional and $F_i$ are type $\tau_i$ functionals, then
\[ \mathcal{V}_\phi(T)(F,F_1, \ldots, F_k)
:= \mathcal{V}_{\phi'}(S)(F_1, \ldots, F_k), \]
where $\phi'(X^\sigma) = F$, but $\phi'$ is otherwise identical to
$\varphi$.
When $T$ has the form $S^{\sigma \rightarrow \tau}R^\sigma$,
\[ \mathcal{V}_\varphi(T)(F_1, \ldots, F_k) =
\mathcal{V}_\varphi(S)(\mathcal{V}_\varphi(R), F_1, \ldots, F_k).\hspace{.2in}\Box \]

It is not hard to show that if $T,S$ are terms such that $T$ is a $\beta$
or $\eta$ redex and $S$ is its contractum, then for all $\phi$,
$\mathcal{V}_\phi(T) = \mathcal{V}_\phi(S)$.

A functional $F$ is \demph{represented} by a term $T$ relative to an assignment $\phi$ if $F = \mathcal{V}_\phi(T)$.

\end{document}